\documentclass[a4paper, 11pt]{article}

\usepackage[english]{babel}
\usepackage[utf8]{inputenc}
\usepackage[T1]{fontenc}
\usepackage{enumerate}
\usepackage{enumitem}
\usepackage[nottoc]{tocbibind}
\usepackage{subcaption} 
\usepackage[title,titletoc]{appendix}

\usepackage{etoolbox}
\usepackage{authblk}
\usepackage{chemarrow}

\usepackage[normalem]{ulem}
\usepackage{stmaryrd}
\usepackage[margin=0.9in]{geometry}

\usepackage{amsmath}
\usepackage{dsfont}
\usepackage{shuffle}
\numberwithin{equation}{section}
\usepackage{graphicx}
\usepackage[colorinlistoftodos]{todonotes}
\usepackage[colorlinks=true, allcolors=blue]{hyperref}
\usepackage{bussproofs}
\usepackage{turnstile}
\usepackage{amsthm}
\usepackage{amssymb}
\usepackage{tabularx}
\usepackage[section]{placeins}
\usepackage{xcolor}
\usepackage{multirow}
\usepackage{multicol}
\usepackage[many]{tcolorbox}
\tcbuselibrary{theorems}
\usepackage{array}
\usepackage{colortbl}
\usepackage{placeins}
\usepackage{titlesec}
\usepackage{mdframed}
\usepackage{nicefrac}
\usepackage{listings}
\usepackage{etoolbox}
\usepackage{mathtools}
\usepackage[intoc, english]{nomencl}
\makenomenclature{}
\usepackage{tikz}
\usetikzlibrary{shapes, arrows.meta, positioning}
\usepackage{todonotes}
\usepackage{algorithm}
\usepackage{algpseudocode}
\usepackage{physics,amsmath}
\usepackage[skip=10pt plus1pt, indent=20pt]{parskip}
\usepackage{fancyhdr}
\usepackage{atbegshi}
\usepackage{thmtools,thm-restate}

\makeatletter
\def\algbackskip{\hskip-\ALG@thistlm}
\makeatother

\makeatletter
\def\thanks#1{\protected@xdef\@thanks{\@thanks
        \protect\footnotetext{#1}}}
\makeatother

\definecolor{dgreen}{RGB}{0,150,0}

\newcommand{\bluebf}[1]{\textcolor{ceruleanblue}{{\bf{#1}}}}

\let\Oldsection\section{}
\renewcommand{\section}{\FloatBarrier\Oldsection}

\let\Oldsubsection\subsection{}
\renewcommand{\subsection}{\FloatBarrier\Oldsubsection}

\let\Oldsubsubsection\subsubsection{}
\renewcommand{\subsubsection}{\FloatBarrier\Oldsubsubsection}
\newtheorem{theorem}{Theorem}[section]
\newtheorem{corollary}{Corollary}[theorem]
\newtheorem{proposition}[theorem]{Proposition}
\newtheorem{lemma}[theorem]{Lemma}

\theoremstyle{definition}
\newtheorem{definition}[theorem]{Definition}

\theoremstyle{remark}
\newtheorem{remark}{Remark}[section]

\theoremstyle{definition}
\newtheorem{example}{Example}[section]

\newtcbtheorem[use counter*=theorem]{definitionBoxBase}{Definition}{}{def}
\newtcbtheorem[use counter*=theorem]{claimBoxBase}{Claim}{}{def}
\newtcbtheorem[use counter*=theorem]{lemmaBoxBase}{Lemma}{}{lem}
\newtcbtheorem[use counter*=theorem]{figureBoxBase}{Figure}{}{fig}
\newtcbtheorem[use counter*=theorem]{theoremBoxBase}{Theorem}{}{thm}
\newtcbtheorem[use counter*=theorem]{propositionBoxBase}{Proposition}{}{prop}

\hypersetup{citecolor=blue}

\usepackage{amsmath,etoolbox}
\usepackage{commath}
\usepackage{amsthm}
\usepackage{amssymb}
\usepackage{bbm}
\usepackage{amsfonts}
\usepackage{enumitem} 
\AtBeginEnvironment{pmatrix}{\setlength{\arraycolsep}{8pt}}

\definecolor{steelblue}{rgb}{0.27, 0.51, 0.71}
\definecolor{mediumelectricblue}{rgb}{0.01, 0.31, 0.59}
\definecolor{halayaube}{rgb}{0.4, 0.22, 0.33}
\definecolor{eggplant}{rgb}{0.38, 0.25, 0.32}
\definecolor{darkraspberry}{rgb}{0.53, 0.15, 0.34}
 \definecolor{bluepigment}{rgb}{0.2, 0.2, 0.6}
 \definecolor{spirodiscoball}{rgb}{0.06, 0.75, 0.99}
 \definecolor{vividburgundy}{rgb}{0.62, 0.11, 0.21}
 \definecolor{mediumtealblue}{rgb}{0.0, 0.33, 0.71}
 \definecolor{mediumvioletred}{rgb}{0.78, 0.08, 0.52}
 \definecolor{egyptianblue}{rgb}{0.06, 0.2, 0.65}
 \definecolor{dodgerblue}{rgb}{0.12, 0.56, 1.0}
 \definecolor{ceruleanblue}{rgb}{0.16, 0.32, 0.75}
 \definecolor{darktangerine}{rgb}{1.0, 0.66, 0.07}
 \definecolor{charcoal}{rgb}{0.21, 0.27, 0.31}
 \definecolor{persianindigo}{rgb}{0.2, 0.07, 0.48}
 \definecolor{orangepeel}{rgb}{1.0, 0.62, 0.0}
 \definecolor{awesome}{rgb}{1.0, 0.13, 0.32}
 \definecolor{green(ncs)}{rgb}{0.0, 0.62, 0.42}

\usepackage{hyperref}
\usepackage{setspace}
\hypersetup{
    colorlinks,
    citecolor=mediumtealblue,
    filecolor=black,
    linkcolor=mediumvioletred,
    urlcolor=darkraspberry 	
}
\usepackage{caption}
\usepackage{algorithm}
\usepackage{algpseudocode}
\usepackage{theoremref}

\DeclareMathAlphabet{\mathpzc}{OT1}{pzc}{m}{it}

\newcommand{\Ex}{\mathbb{E}}

\newcommand{\N}{\mathbb{N}}
\renewcommand{\P}{\mathbb{P}}

\newcommand{\R}{\mathbb{R}}

\newcommand{\Z}{\mathbb{Z}}

\newcommand{\bX}{\mathbb{X}}

\newcommand{\A}{\mathcal{A}}

\newcommand{\F}{\mathcal{F}}

\renewcommand{\L}{\mathcal{L}}

\newcommand{\X}{\mathcal{X}}

\newcommand{\W}{\mathcal{W}}

\DeclareMathOperator*{\argmax}{argmax}

\let\norm\relax
\DeclarePairedDelimiter{\norm}{\lVert}{\rVert}
\DeclarePairedDelimiterX{\inp}[2]{\langle}{\rangle}{#1, #2}
\usepackage[
backend=biber,
style=alphabetic,
sorting=nyt
]{biblatex}
\usepackage{csquotes}
\definecolor{gray75}{gray}{0.75}
\titleformat{\chapter}[hang]{\Huge\bfseries}{\thechapter\hspace{20pt}\textcolor{gray75}{|}\hspace{20pt}}{0pt}{\Huge\bfseries}

\titlespacing*{\chapter}      {0pt}{0pt}{25pt}

\makeatletter
\def\ps@myheadings{%
    \let\@oddfoot\@empty\let\@evenfoot\@empty
    \def\@evenhead{\thepage\hfil\slshape\leftmark}%
    \def\@oddhead{{\slshape\rightmark}\hfil\thepage}%
    \let\@mkboth\@gobbletwo
    \let\sectionmark\@gobble
    \let\subsectionmark\@gobble
    }
  \if@titlepage
  \renewcommand\maketitle{\begin{titlepage}%
  \let\footnotesize\small
  \let\footnoterule\relax
  \let \footnote \thanks
  \null\vfil
  \vskip 60\p@
  \begin{center}%
    {\LARGE \@title \par}%
    \vskip 3em%
    {\large
     \lineskip .75em%
      \begin{tabular}[t]{c}%
        \@author
      \end{tabular}\par}%
      \vskip 1.5em%
    {\large \@date \par}
  \end{center}\par
  \@thanks
  \vfil\null
  \end{titlepage}%
  \setcounter{footnote}{0}%
}
\else
\renewcommand\maketitle{\par
  \begingroup
    \renewcommand\thefootnote{\@fnsymbol\c@footnote}%
    \def\@makefnmark{\rlap{\@textsuperscript{\normalfont\@thefnmark}}}%
    \long\def\@makefntext##1{\parindent 1em\noindent
            \hb@xt@1.8em{%
                \hss\@textsuperscript{\normalfont\@thefnmark}}##1}%
    \if@twocolumn
      \ifnum \col@number=\@ne
        \@maketitle
      \else
        \twocolumn[\@maketitle]%
      \fi
    \else
      \newpage
      \global\@topnum\z@   
      \@maketitle
    \fi
    \thispagestyle{plain}\@thanks
  \endgroup
  \setcounter{footnote}{0}%
}
\makeatother

\def\keywordname{{\bfseries \emph{Keywords}}}%
\def\keywords#1{\par\addvspace\medskipamount{\rightskip=0pt plus1cm
\def\and{\ifhmode\unskip\nobreak\fi\ $\cdot$
}\noindent\keywordname\enspace\ignorespaces#1\par}}
\date{today} 
\thanks{The authors would like to thank William F. Turner for helpful comments as well as Johannes Muhle-Karbe, Cristopher Salvi and Joseph Mulligan for fruitful discussions throughout the writing of this paper. OF and BH thankfully acknowledge the funding of this research by Atlantic House Investments \& EPSRC Oxford-ICL Centre of Doctoral Studies in Mathematics of Random Systems.}

\begin{document}

\pagenumbering{roman}
\setcounter{page}{1}

\author[1]{\Large Owen Futter}
\author[2,3]{Blanka Horvath}
\author[4]{Magnus Wiese}

\affil[1]{\normalsize Imperial College London, Department of Mathematics}
\affil[2]{\normalsize University of Oxford, Mathematical Institute and Oxford Man Insititute}
\affil[3]{\normalsize The Alan Turing Institute}
\affil[4]{\normalsize University of Kaiserslautern, Department of Mathematics}


\title{\vspace{-0.8cm} \endgraf\rule{\textwidth}{.1pt} \\ \tiny \text{ } \\ \vspace{0.3cm} \textbf{\LARGE  Signature Trading: A Path-Dependent Extension of the \\
\vspace{0.15cm} Mean-Variance Framework with Exogenous Signals} \vspace{0.5cm} \\ \endgraf\rule{\textwidth}{.1pt}} 

\date{ }

\pagenumbering{arabic}
\setcounter{page}{1}

\maketitle

\vspace{-0.3cm}
\begin{abstract}
    \noindent In this article we introduce a portfolio optimisation framework, in which the use of \textit{rough path signatures} \cite{Lyons1998DifferentialSignals.} provides a novel method of incorporating path-dependencies in the joint signal-asset dynamics, naturally extending traditional factor models, while keeping the resulting formulas lightweight, tractable and easily interpretable.
    Specifically, we achieve this by representing a trading strategy as a linear functional applied to the \textit{signature of a path} (which we refer to as \emph{``Signature Trading''} or \emph{``Sig-Trading''}). This allows the modeller to efficiently encode the evolution of past time-series observations into the optimisation problem.
    In particular, we derive a concise formulation of the dynamic mean-variance criterion alongside an explicit solution in our setting, which naturally incorporates a drawdown control in the optimal strategy over a finite time horizon.
    Secondly, we draw parallels between classical portfolio stategies and \textit{Sig-Trading} strategies and explain how the latter leads to a pathwise extension of the classical setting via the \textit{``Signature Efficient Frontier''}.
    Finally, we give explicit examples when trading under an exogenous signal as well as examples for momentum and pair-trading strategies, demonstrated both on synthetic and market data.
    Our framework combines the best of both worlds between classical theory (whose appeal lies in clear and concise formulae) and between modern, flexible data-driven methods (usually represented by ML approaches) that can handle more realistic datasets.
    The advantage of the added flexibility of the latter is that one can bypass common issues such as the accumulation of heteroskedastic and asymmetric residuals during the optimisation phase. 
    Overall, \textit{Sig-Trading} combines the flexibility of data-driven methods without compromising on the clarity of the classical theory and our presented results provide a compelling toolbox that yields superior results for a large class of trading strategies.

\end{abstract}

\keywords{\small Mean-Variance Optimisation \and Signature Methods \and Data-Driven Methods \and Dynamic Trading Strategies \and Path-Dependent Signals \and Statistical Arbitrage \and Momentum Strategies \and Stochastic Filtering}
{
  \hypersetup{linkcolor=mediumtealblue}
  \tableofcontents
}

\section{Introduction}
The design and construction of trading strategies is a fundamental aspect of finance and subsequently, has been extensively researched in the past decades. Creating a trading strategy can mainly be divided into two areas: extracting \textit{alpha} and allocating the associated \textit{risk}. Methods of extracting alpha depend on the trade horizon or trade frequency and are often achieved through statistical techniques, which can then be incorporated into a parametric model. Meanwhile, allocating the associated risk generally depends on your choice of model and objective criterion; conventionally, optimisation methods are deployed to find a transformation of the signal-asset dynamics that maximises the chosen utility function of the strategy PnL, with respect to the underlying model. The work  (\cite{Markowitz1952PortfolioSelection}) of Markowitz introduced modern portfolio theory based on portfolio allocation determined by investors' preferences to risk and returns, resulting in a well-diversified portfolio. Classical methods consist of fixing a class of probabilistic parametric models and calibrating the model's parameters with respect to the empirical price process. Depending on the choice of parametric model and its corresponding parameters, the optimal portfolio will be different. In well-studied families of models it is possible to find closed-form expressions of the optimal portfolio given certain risk and return preferences, however in more complex parametric models or with more general utility functions, explicit solutions may not be attainable and so numerical techniques are used. \par

\begin{tikzpicture}
    \centering
    
    \hspace{-0.94cm}
    \draw[rounded corners=5pt] (-2.4, 0) rectangle (1.9, 1);
    \node at (-0.25, 0.5) {\parbox{4.5cm}{\centering Observe Signal \& Asset}};
    
    \draw[->] (1.9, 0.5) -- (3.6, 0.5);
    \node at (2.75, 0.9) {Model};
    
    \draw[rounded corners=5pt] (3.6, 0) rectangle (7.8, 1);
    \node at (5.7, 0.5) {\parbox{4cm}{\centering Characterise Dynamics}};

    \draw[->] (7.8, 0.5) -- (9.7, 0.5);
    \node at (8.75, 0.9) {Optimise};

    \draw[rounded corners=5pt] (9.7, 0) rectangle (14.1, 1);
    \node at (11.9, 0.5) {\parbox{6cm}{\centering Maximal Utility Strategy}};
\end{tikzpicture}

We can observe from the above blueprint, that the possible choices of model are vast - for example all possible ways of feature engineering, return prediction, or the choice of probabilistic model. The choice of optimisation technique will then be tailored to the choice of utility function and the model. While well-studied models provide a useful benchmark for practitioners, many of the assumptions made are restrictive. Several modelling assumptions do not reflect stylised facts in practice such as non-stationarity, heavy tails and path-dependent volatility (\cite{Fukasawa2021VolatilityRough, Das2022RoughnessSignals,
Guyon2022VolatilityPath-Dependent, Morel2023PathMonte-Carlo}) - all of which are exhibited by financial time-series data (\cite{Cont2001EmpiricalIssues}). The latter considerations are even more relevant when working with multiple assets since dependence structures are highly non-linear, often with cross-sectional trend- and mean-reversion patterns  occurring in practice. 
Recent research activity has brought modelling approaches to the forefront that are inherently data-driven and provide a highly flexible framework that is able to capture a broader set of asset dynamics than currently used classical stochastic models did. Data-driven techniques for trading have increased in recent years due to the rise of machine learning applications in finance, such as in \cite{Zhang2020DeepTrading, Wang2019PortfolioLearning, Zhang2020DeepOptimization, Jaimungal2021RobustLearning, VanStaden2021ACosts, Pretorius2022DeepManagement, Guijarro-Ordonez2021DeepArbitrage, Coache2022ConditionallyLearning}, as well as in hedging applications in \cite{Buhler2018DeepHedging, Horvath2021DeepVolatility, Limmer2023RobustGANs}. The requirement for complex and accurate synthetic data to train and test trading frameworks has also led to extensive research in developing market generators (\cite{Buehler2021GeneratingSignatures, Wiese2019QuantSeries, Ni2020ConditionalGenerationb, Ni2021Sig-WassersteinGeneration, Issa2023Non-adversarialScores}).
In this work, we go a step further than just modelling asset dynamics in a model-free way and take the approach of utilising rough path theory (\cite{Lyons2014RoughStreams, Friz2010MultidimensionalPaths}) also to develop a trading strategy that is determined by the expected signature of the joint signal-asset process. Since the aforementioned expected signature uniquely determines the law of the stochastic process (\cite{Chevyrev2013CharacteristicPaths}), it is immediate to see how the techniques simplify to the classical case when applied to traditional stochastic models. By taking a pathwise approach to trading, this naturally alleviates probabilistic restrictions, resulting in a \textit{model-free} or model-agnostic setup (\cite{Chiu2023AFinance}). A pathwise setting has been used in \cite{Perkowski2016PathwiseFinance, Riga2016ATrading, Riga2016PathwiseFinance, Ananova2023Model-freeStrategies}, and has more recently been utilised to handle more general portfolios in \cite{Allan2021Model-freeApproach}, as well as in applications to derivative pricing and calibration in \cite{Cuchiero2022Signature-basedCalibration} and optimal stopping problems in \cite{Bayer2021OptimalSignatures}. \par
 Inspired by the work of Perez et al. in \cite{Arribas2018DerivativesPayoffs, Kalsi2020OptimalSignatures, Lyons2019NumericalSignatures}, we adopt the idea of representing a trading strategy as a linear functional applied to the signature of a path and extend this to to incorporate exogenous market signals and multiple assets. We can think of a trading strategy as a decision made using the knowledge of the current state (e.g. the previous price path, plus some exogenous trading signal); in other words as a function from path space to some decision process. In practice, the true driving processes are most likely not (directly) observable and hidden by layers of noise. Hence the mapping from the previous price process to a trading strategy is often done first by removing noise in the system through stochastic filtering \cite{Bain2009FundamentalsFiltering, Crisan2021PathwiseProblem} (i.e. an exponentially weighted moving average or Kalman filter) and then optimised based on the resulting prediction.  Recently, the authors in \cite{Cohen2023NowcastingMethods} prove how the Kalman filter can in fact be equivalently written as a linear regression on the signature. Work has also been conducted in \cite{Dyer2021DeepModels, Dyer2022ApproximateDiscrepancies} with relation to approximate Bayesian computation using signatures. This naturally prompts the question - can the class of linear functionals on the signature be seen as a rich enough class of maps that represent such trading strategies? We will show that this is possible to due to signatures’ capability to approximate continuous functions on paths. Also perhaps crucially, the Sig-Trading framework does not impose the restriction that the underlying asset or signal be Markovian, allowing the Sig-Trader to capture auto-correlation and mean-reverting behaviours within the process that perhaps some more classical methods are not able to exploit. This enables us to incorporate path-dependent considerations into our trading decisions, while still obtaining a closed-form solution, that is easy to compute and simple to analyse. 

\noindent The paper is organised as follows: In Section~\ref{sec:market_model} we introduce key foundations of Sig-Trading and the concept of the extension to classical factor models. In Section~\ref{sec:original} we present our main result, an analytic solution to the dynamic mean-variance criterion for Sig-Trading and compare this to original factor models, whilst introducing a Sig-Trading version of the efficient frontier. Finally, in Section~\ref{sec:implement} we discuss its implementation and in Section~\ref{sec:results} we highlight the advantages of Sig-Trading, provide intuitive examples and demonstrate its possible use cases in practice, such as pairs trading, momentum, and trading under an exogenous signal. 
\subsection{Background and Motivation}

\subsubsection*{The Objective}

In this article, we are concerned with finding an optimal \textit{systematic} and \textit{dynamic} trading strategy such that the trader continuously updates their position as new information filters in. This is done with no discretion and is defined as a function of the market state, which continuously updates through time, leading to a new position for each time $t$. \par

We denote $T \in \R$ the terminal time. Let $(\Omega, \F, (\F_t)_{t \in [0,T]}, \P)$ be a filtered probability space. Furthermore, denote by $X=(X_t)_{t \in [0,T]}$ a non-negative $\R^d$-valued stochastic process satisfying $X_0^m=1$ for $m \in \{1,\dots,d\}$. We are interested in finding an optimal, predictable dynamic trading strategy $(\xi)_{t \in [0,T]}$ that maximises the expected utility of the PnL
\begin{align}{\label{eq:optimisation}}
    \max_{\substack{(\xi_t)_{t \in [0,T]} \\ \text{s.t constraints}}} \Ex(U(V_T)) 
\end{align}
where $V_T$ is the terminal value of the trading strategy (i.e the PnL)
\begin{align}{\label{eq:PnL}}
    V_T = \sum_{m=1}^d \int_0^T \xi_t^m dX_t^m.
\end{align}
Here, the optimal strategy $\xi$ is a function of the market state (i.e filtration $\F$ at time $t$) and its optimality is with respect to the constraints and utility function that the trader chooses. In order to solve such an optimisation problem, commonly methods are constructed from the following structure:
\begin{enumerate} 
\itemsep-0.5em
    \item Fix a framework to model the underlying dynamics of $X$,
    \item Choose an objective criterion (utility and constraints),
    \item Optimisation with respect to (1) and (2).
\end{enumerate}
Most likely, a trader may be trading under the presence of exogenous information $f$ to enrich the filtration $\F$ (and hence the model of the dynamics of $X$). Recently, extensive research in this direction has been conducted in optimal execution literature, e.g. in a more classical setting in \cite{Lehalle2019IncorporatingTrading, Forde2022OptimalResilience, Sanchez-Betancourt2022BrokersSignals, Cartea2022ExecutionMakers, Bank2023OptimalSignals, Cont2023FastInformation} and with machine learning applications in \cite{Bergault2021Multi-assetDynamics, Cartea2023BanditsSignals}. In this work, we do not consider market impact, but refer the reader to \cite{Kalsi2020OptimalSignatures, Cartea2022Double-ExecutionSignatures} to work in this area involving signatures. 

\subsubsection*{Modelling Dynamics}

The choice of model in (1) often requires explicitly predicting asset returns through supervised learning techniques. Therefore, a large number of models tend to fall into the \textit{predict-then-optimise} framework, where heavy assumptions are made on the asset returns/market factors that are input into any prediction, leaving the final solution exposed to asymmetric and compounded errors. Solutions to these problems are very well studied with specific assumptions and restrictions on the underlying asset process, however without these assumptions this can be much more difficult. 
In such stochastic control problems, the asset dynamics are explained by a diffusion process and dynamic programming can then be used to solve the Hamilton-Jacobi-Bellman (HJB) equation. In practice, the future expected returns (the drift of the process $X$), $\mu_{t+1}$, are often predicted via supervised learning methods, using trading signals or \textit{factors} as statistical predictors which are embedded into the framework itself (\cite{Chamberlain1983ArbitrageMarkets, Ng1992AReturns, Fama1993CommonBonds, Fama2015AModel, Stock2010DynamicModels, Garleanu2013DynamicCosts}). A generic (linear) factor model models the asset returns at time $t$, as
\begin{align} \label{eq:factor_model_predict}
    \mu_{t+1} = \Ex [ r_{t+1} \vert \F_t ] = Bf_t + \varepsilon_{t+1}
\end{align}
where $r$ is the $d$-dimensional asset returns, $B$ is a $d \times N$ matrix of factor coefficients, $f_t$ is a vector of $N$ factor returns and $\varepsilon$ is a vector of the $d$ assets’ (unexplained) residuals returns. This is set up as a supervised linear regression on future returns, as a function of the trading signals/factors. \par

However, when working with financial data, there is generally a very low signal to noise ratio and the residual terms $\epsilon_t$ are badly behaved, violating many statistical assumptions. Due to autocorrelation, non-stationarity and path-dependent volatility in the underlying asset $X$ (and also the signal $f$), this framework can very quickly become problematic, and these asymmetric errors are then compounded in the optimisation phase. In order to capture some of the autocorrelation in the residuals, a trader could incorporate the path into the signal $f$ via stochastic filtering (\cite{Bain2009FundamentalsFiltering, Crisan2021PathwiseProblem}). Traders may also scale returns for volatility in order to remove heteroskedasticity (\cite{Engle2001GARCHEconometrics}), log transform to remove asymmetricity or winsorise to remove fat tails; in \cite{Das2022RoughnessSignals}, the roughness of signals is also discussed. However, there still remains a large amount of discretion in such feature engineering and we will show that the signature can be used efficiently and robustly to tackle these issues in a data-driven manner. To avoid the accumulation of mis-specified error terms from the prediction phase, end-to-end (E2E) approaches using machine learning frameworks have been used in \cite{Costa2022DistributionallyConstruction} and \cite{Zhang2021ALearning} to ensure robustness by bypassing the prediction stage.

\begin{figure}[ht]
\begin{tikzpicture}
    \centering
    \hspace{-0.8cm}
    \draw[rounded corners=5pt] (-2.4, 0) rectangle (0.8, 1.6);
    \node at (-0.8, 0.8) {\parbox{4.5cm}{\centering Signal \& Asset \\ $f_t, X_t$}};
    
    \draw[->] (0.8, 0.8) -- (3.7, 0.8);
    \node at (2.25, 1.35) {\parbox{4cm}{\centering Least Squares Prediction}};
    
    \draw[rounded corners=5pt] (3.7, 0) rectangle (7.7, 1.6);
    \node at (5.7, 0.8) {\parbox{5cm}{\centering Future Returns \\ $\pi(f_t,X_t) = \Ex (X_{t+1} \vert f_t)$}};

    \draw[->] (7.7, 0.8) -- (10.6, 0.8);
    \node at (9.15, 1.2) {Optimisation};

    \draw[rounded corners=5pt] (10.6, 0) rectangle (14.1, 1.6);
    \node at (12.3, 0.8) {\parbox{4cm}{\centering Strategy \\ $\xi_t = \phi^* (\pi(f_t,X_t))$}};

    \draw[mediumtealblue,->] (-0.8, 0) -- (-0.8, -1);
    \draw[mediumtealblue,rounded corners=5pt] (-2.4, -1) rectangle (0.8, -2.6);
    \node at (-0.8, -1.8) {\parbox{4.5cm}{\centering Feature Space \\ $\varphi(f_t,X_t)$}};

    \draw[mediumtealblue,->] (0.8, -1.8) -- (10.6, -1.8);
    \node at (5.7, -1.4) {Optimisation};

    \draw[mediumtealblue,rounded corners=5pt] (10.6, -1) rectangle (14.1, -2.6);
    \node at (12.35, -1.8) {\parbox{4cm}{\centering Strategy \\ $\xi_t = \varphi^*(f_t,X_t))$}};
\end{tikzpicture}
  \caption{\centering Original framework (above) and end-to-end (E2E) optimisation framework (below).}
  \label{fig:end_to_end}
\end{figure}
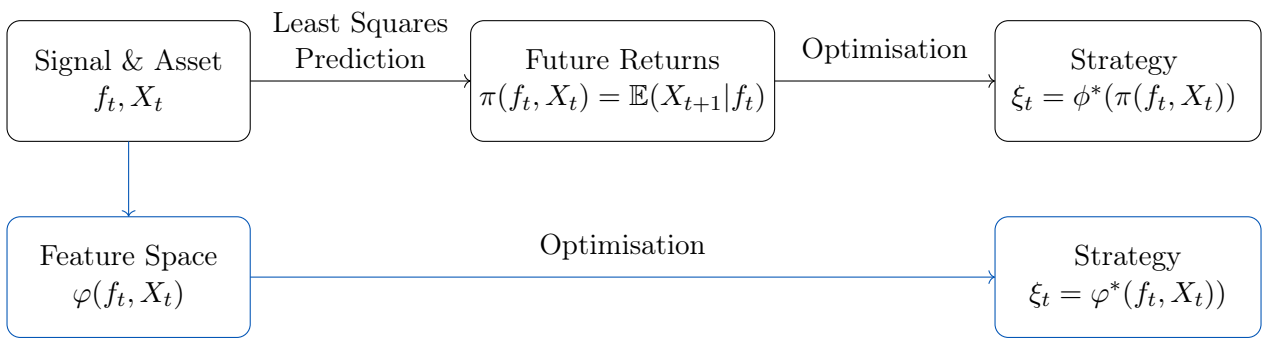

\subsubsection*{Choice of Utility}

Once we have a model that characterises the dynamics of the driving signal and the underlying asset, we can proceed to transforming this into a trading strategy position. How one does this depends on a variety of conditions such as the type of strategy, if we are trading multiple assets, risk preferences, trade frequency and investment horizon. 
Common objective criteria focus on a single trade-by-trade optimisation basis, overlooking the potential path that the trading strategy will take. However as alluded to previously, in practice signals and underlying assets can have strong temporal dependencies and so subsequent trading strategy positions will inherit autocorrelation structure. \par
\begin{figure}[ht]
\centering
  \includegraphics[width=\linewidth]{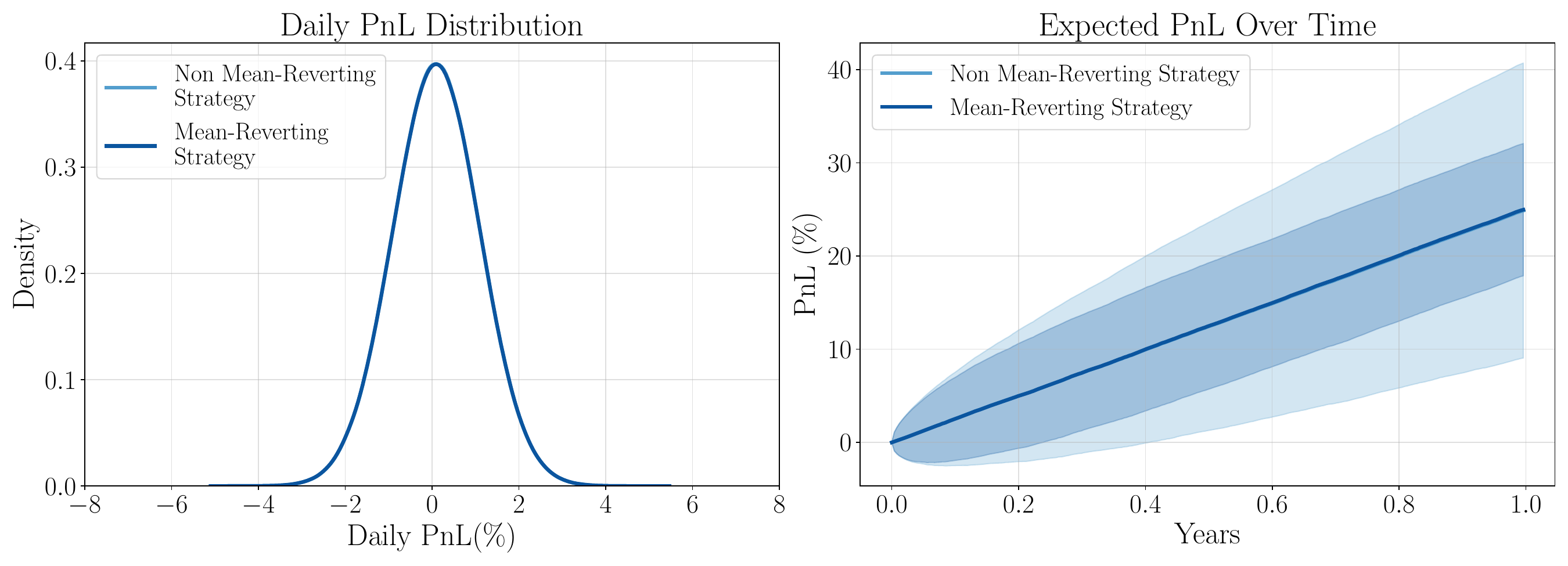}
  \caption{\centering Comparison of the PnL profile through time of a mean-reverting strategy vs a non mean-reverting strategy.}
  \label{fig:dynamic_strategy}
\end{figure}
As a motivating example in Figure~\ref{fig:dynamic_strategy}, we consider the comparison between a strategy that is mean-reverting (has autocorrelation) vs one that doesn't. Both strategies yield identical daily PnL distributions (and hence Sharpe ratio), but have different distributions at a future time due to the temporal structure within the strategy over time. By optimising with respect to a future time horizon, we are inherently optimising for a dynamic criterion that is dependent on the path. This approach naturally integrates a drawdown control in the optimisation, since drawdowns are a path-dependent characteristic. Achieving a strategy with a small maximum drawdown either requires a strong trend-to-noise ratio (high Sharpe ratio) or relies on the strategy being mean-reverting during volatile periods (\cite{Rej2017YouWorrying} provides a neat analysis). \par

Mean-variance optimisation, perhaps the most widely known choice of utility function, was first introduced in Markowitz's thesis (\cite{Markowitz1952PortfolioSelection}). He demonstrated that it was natural to construct an objective function that rewarded positive returns while penalizing associated risk (via variance). Subsequently, this framework has been researched extensively in the literature, including the CAPM asset pricing framework (\cite{Sharpe1964CapitalRisk}). In this work, we extend this approach to integrate path-dependencies and optimise under a \textit{dynamic} mean-variance criterion. We do so by simultaneously capturing path-dependent structure in the dynamics, while managing path-dependent variance through the lifetime of the trade. The general dynamic mean-variance optimisation can be framed as
\begin{align}{\label{eq:mean_var_optimisation}}
    \max_{(\xi_t)_{t \in [0,T]}} \Ex \left(\sum_{m=1}^d \int_0^T \xi_t^m dX_t^m \right) - \frac{\lambda}{2} \textup{ Var}\left(\sum_{m=1}^d\int_0^T \xi_t^m dX_t^m \right),
\end{align}
where the quantity
\begin{align*}
    V_T \coloneqq \sum_{m=1}^d \int_0^T \xi_t^m dX_t^m 
\end{align*}
is the strategy PnL at some future time $T$.

\subsubsection*{Optimisation}

 Depending on the choice of underlying model, the optimisation method can be different - for example, reinforcement learning can in fact be used to simulataneously learn the model and the optimal strategy (\cite{Wang2019Continuous-TimeFramework, Brini2023DeepReturns, Sood2023DeepOptimization}). The authors in \cite{Kalayci2019AOptimization} provide an extensive overview of alternate formulations and methods for mean-variance optimisation. In the case of the predict-then-optimise framework, such as \eqref{eq:factor_model_predict}, the mean-variance solution is given by
\begin{align} \label{eq:mean_var_sol}
    \xi_t^* = \frac{1}{\lambda} \Sigma_t^{-1} \mu_t
\end{align}
where $\Sigma_t$ is the $d \times d$ covariance matrix of asset returns at time $t$ and $\mu_t$ is the $d \times 1$ vector of expected returns over the next time period $t \in [0,1]$. This solution is highly intuitive and tractable. The performance of such a trading strategy is then mostly reliant on the predictive power of the signal that goes into the model, as well as the well-posedness of the covariance structure, leaving flexibility and responsibility in the hands of the trader. However, as has been highlighted in decades of literature, this framework is highly restrictive, such as the assumption of normally distributed returns and stationarity in the factor signal.

\subsection{Summary of Contributions} \label{section:our_approach}

Several stylised facts exhibited in asset prices are not represented in most classical dynamic optimisation frameworks. Many of these stylised facts, such as slow decaying autocorrelation and volatility clustering, are path-dependent properties and so it advantageous to incorporate path-dependence into any factor model. Due to such characteristics in financial time series data, the predict-then-optimise composition in classical factor models can be prone to asymmetric errors that are compounded through the optimisation process. Work from \cite{Costa2022DistributionallyConstruction} and \cite{Zhang2021ALearning} have provided robust machine learning approaches to help account for this issue by bypassing the direct prediction phase. Alternatively, we are able to exploit such dynamics in the signal and asset and incorporate these not only to increase expected return, but to dynamically reduce variance of PnL. Due to the powerful properties arising in rough path theory we are able to simultaneously incorporate such dynamics into both the modelling framework and the optimisation phase at the same time, bypassing any explicit prediction.

In summary, our contribution is to extend the existing signature trading strategy framework first presented in the thesis of Perez (\cite{PerezArribas2020SignaturesFinance}), by deriving a closed form mean-variance optimal trading strategy for multiple assets that accounts for path-dependent dynamics between exogenous trading signals and the underlying assets, thereby providing a pathwise extension to classical factor models. Our framework is simple to implement and does not require heavy machinery while comparative machine learning methods may require extensive network building and hyperparameter tuning every time the trader wants to perform a new optimisation. This is possible due to the mathematical properties of the signature allowing us to linearise the objective function, meaning it is relatively straightforward to solve for an explicit closed-form solution. Sig-Trading is a one-model-fits-all type framework that is flexible enough to adapt to any type of underlying asset and market factor process. Once a linear functional $\ell$ is obtained from past data samples it is straightforward to unravel this into an implementable trading strategy characterised by the number of units to buy/sell at each time point $t \in [0,T]$. Since the strategy is dynamic, as new data arrives the Sig-Trader will continuously compute the signature and update their position accordingly.

\section{
The Modelling Setup of the Sig-Factor Model
}\label{sec:market_model}

Rough path theory has provided many valuable tools, such as the Signature of a path, to help shift the focus from probabilistic to pathwise approaches when working with streams of data. The concept of the signature was first introduced in \cite{Chen1957IntegrationFormula, Chen1977IteratedIntegrals} and has played a crucial role in rough path theory in \cite{Lyons1998DifferentialSignals., Lyons2007DifferentialPaths, Friz2010MultidimensionalPaths, Friz2020AStructures}. Recent mathematical finance literature has benefited immensely from its universality property, which states that linear functionals on the signature are dense in the space of continuous functions on compact sets of paths (Theorem~\ref{eq:universal_approx}). This result allows to approximate a trading strategy as a linear functional on the terms of the signature. Incorporating higher order path-dependent characteristics via the signature allows us to capture stylised facts of financial time series data, without requiring or imposing an explicit probability distribution on the future returns. In this section, we discuss the notion of a $\textit{Signature Trading Strategy}$, first introduced in \cite{Lyons2019NumericalSignatures}, and how this is incorporated into a $\textit{model-free}$ setting, as well as under the presence of exogenous market signals. 

\subsection{The Signature} \label{sec:sig_intro}

In this section, we first recall definitons of path augmentations and how these are used in preceding results such as Theorem~\ref{eq:hoff_converegence}, which is crucial for our main result Theorem~\ref{thm:orig_solution}. Whilst we introduce fundamental definitions in Section \ref{sec:sig_intro}, we have collected some basic concepts and definitions from rough path theory for convenience in Appendix \ref{sec:appx_rough_paths} as they may aide in understanding of notations and technicalities throughout the paper. For a more thorough introduction to signatures, see for example \cite{Gyurko2013ExtractingStream, Chevyrev2016ALearning, Levin2013LearningSystem, Fermanian2021EmbeddingSignatures, Lyons2022SignatureLearning} for excellent articles focusing on intuition and understanding in a practical setting. For specific applications of signatures in finance we refer the reader to \cite{Bonnier2019DeepTransforms, Kalsi2020OptimalSignatures, Arribas2020Sig-SDEsFinance, Bayer2021OptimalSignatures, Cuchiero2022Signature-basedCalibration, Alden2022Model-AgnosticSignatures, Dupire2023FunctionalExpansions, Issa2023Non-parametricStructures, Wiese2023Sig-Splines:Models}.

Unless stated otherwise, the process $(X_t)_{t \in [0,T]}$ is a  continuous, stochastic process defined on a filtered probability space $(\Omega, \F, (\F_t)_{t \in [0,T]}, \P)$. We often refer to the path trajectories of the process $(X_t)_{t \in [0,T]}$ as $X:[0,T] \to \R^d$, which we assume can be lifted to geometric rough paths (Definition~\ref{defn:geom_rough_path}).

\begin{definition}{(Time reparameterisation).}
    Let $X:[0,T] \to \R^d$, $\varphi:[0,T]\to[T_1, T_2]$ a non-decreasing surjection, then the re-parameterised path is denoted as $X \circ \varphi =: X^{\varphi} : [T_1, T_2] \to \R^d$.
\end{definition}

\begin{definition}{(Add-time process).} 
    Often, we may wish to preserve the temporal structure of a path and so we keep the time parameterisation of the path $X$ by defining a new process, $\hat{X}$. We denote the time-augmented process by $\hat{X}_t=(t, X_t), t \in [0,T]$ such that $(t,X_t) =: \hat{X}: [0,T] \to \R^{d+1}$.
\end{definition}

\begin{definition}{(Hoff Lead-Lag Process, \cite{Hoff2006ThePath}, \cite{Flint2016DiscretelyProcess}).} \label{eq:hoff_defn}
    Let $\hat{X}:[0,T] \to \R^{d+1}$ be the continuous time-augmented process of $X$, discretely sampled at $t=t_0,\dots,t_{2N}$. The \textit{Hoff lead-lag transformed path} is defined as the piecewise linear interpolation $\hat{X}^{LL}:[0,T] \to \R^{2(d+1)}$ such that 
    $$
    (\hat{X}_{t_i}^{LL})^{2N}_{i=1} = (\hat{X}_{t_i}^{\textup{lead}}, \hat{X}_{t_i}^{\textup{lag}})^{2N}_{i=1}
    $$ 
    where 
    $$
    \hat{X}_{t}^{\textup{lead}} =
    \begin{cases}
        \hat{X}_{t_{k+1}}, & \quad \text{if } t \in [2k, 2k+1] \\
        \hat{X}_{t_{k+1}} + 2(t-(2k+1))(X_{t_{k+2}} - X_{t_{k+1}}), & \quad \text{if } t \in [2k+1, 2k+\frac{3}{2}] \\
        \hat{X}_{t_{k+2}}, & \quad \text{if } t \in [2k+\frac{3}{2}, 2k+2],
    \end{cases}
    $$
    $$
    \hat{X}_{t_j}^{\textup{lag}} = 
    \begin{cases}
        \hat{X}_{t_i}, & \quad \text{if } t \in [2k, 2k+\frac{3}{2}] \\
        \hat{X}_{t_{k+1}} + 2(t-(2k+\frac{3}{2}))(X_{t_{k+1}} - X_{t_{k}}), & \quad \text{if } t \in [2k+\frac{3}{2}, 2k+2].
    \end{cases}
    $$
    \end{definition}
    \begin{remark}
       There also exists a more intuitive and straightforward definition of a lead-lag process, however we decide not to opt for this version and instead consider the so-called \textit{Hoff process} (\cite{Hoff2006ThePath}), due to its fundamental properties in the case when we want to calculate the PnL of our trading strategy (Theorem \ref{eq:hoff_converegence}). This is due to the powerful, non-trivial result proven in \cite{Flint2016DiscretelyProcess} that states that the It\^{o} integral of a function of a process $X$ against itself, can be recovered via the components of the Hoff lead-lag transformation. This will be made more precise in Section~\ref{sec:sig_factor_model}.
    \end{remark}
\begin{figure}[ht]
\centering
  \includegraphics[width=\linewidth]{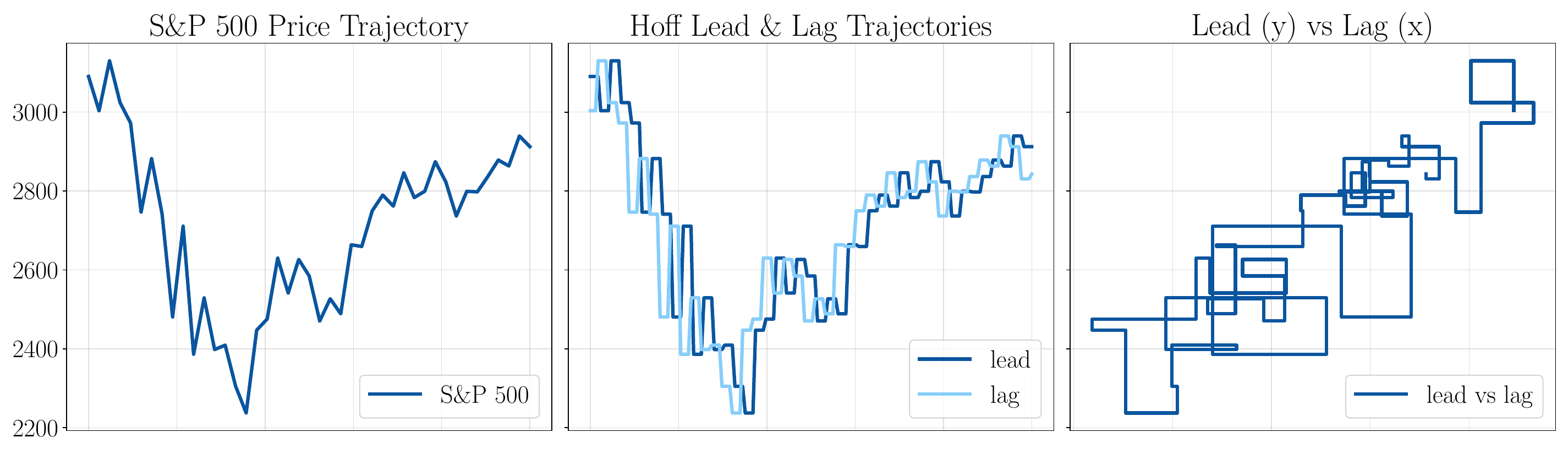}
  \caption{\centering SPY ETF sample price trajectory between 02/03/2020-30/04/2020 (LHS) and its respective Hoff lead-lag transformation (Centre) and the lead vs the lag component (RHS).}
  \label{fig:hoff_lead_lag}
\end{figure}


\noindent \textbf{Notation:} We distinguish that any integral of the form $\int f \circ dx$ refers to Stratnovich integration, meanwhile $\int f  dx$ will refer to It\^{o} integration. 

\begin{definition}{(Signature Transform).}
    Let $X:[0,T] \to \R^d \in C^{1-var}([0,T];\R^d)$ be a (piecewise) smooth path. Let us define the simplex $\Delta_T = \{(s,t) : 0 \leq s \leq t \leq T \}$. The signature of $X$ between fixed time $s$ and $t$ is a map
    \begin{align*}
        \bX : \Delta_T & \to T((\R^d)) \\
        (s,t) & \mapsto \bX_{s,t} := (1, \bX_{s,t}^1, \dots, \bX_{s,t}^n, \dots )
    \end{align*}
    where the $n$-th order of the signature is defined as
    \begin{align*}
        \bX_{s,t}^n := \idotsint\displaylimits_{s < u_1 < \dots < u_n < t} dX_{u_1} \otimes \dots \otimes dX_{u_n} \in (\R^d)^{\otimes n}.
    \end{align*}
        The signature is a $T((\R^d))$-valued process, which can be viewed as
        \begin{align*}
        \bX^{< \infty}_{0,T} = \left( \underbrace{
        \begin{matrix}
            \text{ } \vspace{0.2cm}  \\
            \bX_{0,T}^{\bluebf{\emptyset}} \\
            \text{ } \vspace{0.2cm}
        \end{matrix}}_{\textstyle =1} \text{ } , \text{ }
        \underbrace{\begin{pmatrix}
            \bX_{0,T}^{\bluebf{1}}   \\
            \vdots \\
            \bX_{0,T}^{\bluebf{d}}
        \end{pmatrix}}_{\textstyle \bX_{0,T}^1}
        \text{ } , \text{ }
        \underbrace{\begin{pmatrix}
            \bX_{0,T}^{\bluebf{11}} & \dots &  \bX_{0,T}^{\bluebf{1d}}   \\
            \vdots & \ddots & \vdots \\
            \bX_{0,T}^{\bluebf{d1}} & \dots & \bX_{0,T}^{\bluebf{dd}}
        \end{pmatrix}}_{\textstyle \bX_{0,T}^2} \text{ } , \text{ } \dots
        \right).
    \end{align*}
    The signature of a path can be truncated at any finite order $N \in \N$. We denote the truncated signature up to order $N$ as
    \begin{align*}
        \bX^{\leq N} : \Delta_T & \to T^{(N)}(\mathbb{R}^d) \\
        (s,t) & \mapsto \bX_{s,t} := (1, \bX_{s,t}^1, \dots, \bX_{s,t}^N).
    \end{align*}
\end{definition}
\noindent \textbf{Notation:} Throughout, we refer to the (un-truncated) signature between time $0$ and time $T$ as $\bX^{< \infty}_{0,T}$. We may refer to $\bX_{0,T}^n$ as being the $n$-th \textit{level} of the signature and $\bX_{0,T}^{\leq N}$ as the $N$-th \textit{order} truncated signature.

\noindent \textbf{Notation:} We will use bold blue for any word $\bluebf{w} \in \W(A_d)$. Words are the multi-indices that can be thought of as the different multi-indices inside the tensor algebra. By using a specific word $\bluebf{w}$, this will often equate to referring to the specific multi-index associated to that word.\par

\begin{definition}{(Linear functionals on the tensor algebra).}
    Note that there is a natural pairing between the extended tensor algebra $T((\R^d))$ and its dual space $T((\R^d)^*)$, by which we denote
    \begin{align*}
        \langle \cdot , \cdot \rangle : T((\R^d)^*) \times T((\R^d)) \to \R
    \end{align*}
    and is given by
    $$
    \langle \ell , \bX \rangle = \sum_{\bluebf{w} \in \W(A_d)} \ell_{\bluebf{w}} \bX^{\bluebf{w}}
    $$
    for  $\ell \in T((\R^d)^*)$, $\bX \in  T((\R^d)).$ We make clear that there exists a canonical isomorphism $(\R^d)^* \cong \R^d$ through the mapping $(\R^d)^* \ni \langle \ell , \cdot \rangle \mapsto \ell \in \R^d$.
\end{definition}

\subsection{Sig-Factor Model} \label{sec:sig_factor_model}
In this section, we develop and collect the tools required in order to find optimal trading strategies with respect to a given objective criterion. The key ingredient is that we can approximate a trading strategy as a linear functional on the signature of the path, which is formulated in Theorem \ref{thm:sig_approx}. Then, by embedding the \textit{market factor process} (Definition~\ref{defn:market_factor_process}) as a geometric rough path via the signature and using Theorem \ref{eq:hoff_converegence}, we alleviate most probabilistic restrictions that classical methods may have. 
Theorem~\ref{thm:int_sig} subsequently allows us to in fact bypass the It\^{o} integral entirely and express the expected PnL as a linear functional on the expected Hoff lead-lag signature. Hence, all that remains to be done in order to optimise under a given criterion (i.e mean-variance), is to solve an optimisation problem - equating to solving a system of linear equations. \par
We consider the scenario, as is often the case in practice, that we observe a much larger market state than just the asset trajectories themselves, i.e we observe the natural filtration of a new process that consists of the original price trajectories, enriched with some new market factors $f_t = (f^1_t, \dots, f^N_t)$. A \textit{market factor} can be any un-tradable characteristic, classified as a stochastic process $(f_t)_{t \in [0,T]}$ that may be used in tandem with our framework, in order to enrich the market state and provide predictability. In practice, factors are chosen and understood as a driving signal of the the underlying process. Unless stated otherwise, we assume the underlying asset process $(X_t)_{t \in [0,T]}$ is a  continuous, stochastic process whose path trajectories can be lifted to geometric rough paths (Definition~\ref{defn:geom_rough_path}). We note that, while our price process may be a semi-martingale, we do not require any restriction on the nature of the exogenous market signals, where as traditional factor models may require stationarity.

\begin{definition}{(Market factor process).} \label{defn:market_factor_process}
    Let $X=(X_t)_{t \in [0,T]}$ be a $d$-dimensional tradable underlying asset process and $\{f^i \}_{i=1}^N$ are $N$ un-tradable exogenous trading signals. Then we define the (time-augmented) \textit{market factor process} as the $(1+d+N)$-dimensional process
    \begin{align*}
        \hat{Z}_t \coloneqq (t, X_t, f_t) 
    \end{align*}
    that induces the natural filtered probability sapce $(\Omega, \F^{Z}, (\F^{Z}_t)_{t \in [0,T]}, \P)$, where the filtration of $X$ satisfies $\F^X \subseteq \F^Z$ such that $X$ is driven by $\{f^i \}_{i=1}^N$. Throughout the remainder of this paper, we maintain such assumptions on $\hat{Z}$. We define as the space of all \textit{market factor trajectories} for given assets $X$ and factors $f$.
    $$
    \mathcal{Z}^{f}_{0,T} \coloneqq \{ \text{   } \hat{Z}_t = (t, X_t, f_t) \quad \vert \quad X:[0,T] \to \R^d,  \text{ and } f:[0,T] \to \R^N \text{ and } Z_0 = (0,1,1)  \text{   } \}
    $$
\end{definition}

\begin{definition}{(Exogenous Signature Trading Strategy).} \label{def:exog_sig_strat}
    Let $\hat{Z}$ be a market factor process as defined in Definition~\ref{defn:market_factor_process}. Let $\xi = (\xi)_{t \in [0,T]}$ be an adapted, $\F^Z_t$-predictable and integrable strategy such that $\int^T_0 \xi_s^2 ds < \infty$. If the strategy $\xi$ is then a function of the market state, that is $\xi_t = \phi(\hat{Z}_{0,t})$, a continuous function of the previous \textit{market factor trajectory} up to time $t$. We can extend $\xi$ to be a \textit{signature trading strategy}, such that 
    \begin{align*}
        \xi_t^m = \langle \ell_m, \hat{\Z}_{0,t} \rangle \approx \phi(t,Z_{0,t}), \quad \forall m = 1,\dots,d.
    \end{align*}
    We define the space of all \textit{exogenous signature trading strategies} with respect to market factors $f$ as
    $$
    \A^{f, \text{sig}} \coloneqq \left\{ (\xi)_{t \in [0,T]} = (\langle \ell, \hat{\Z}_{0,t} \rangle)_{t \in [0,T]} \quad \bigg\vert \quad \int^T_0 \xi_s^2 ds < \infty, \quad \forall \text{ } \hat{Z}_{0,T} \in \hat{\mathcal{Z}}^f_{0,T} \text{ and } \ell \in T((\mathbb{R}^{N+d+1})^*) \right\}.
    $$ 
\end{definition}
\begin{remark}
    If instead we wanted to trade \textit{endogenously} without any trading signals, we would simply take the market factors to be the null-process $\emptyset$, such that $\hat{X} = \hat{Z}$, the results would still hold.
\end{remark}
Now that we have defined the framework required in order to trade a signature trading strategy embedded with market factors, we can state a key result in order to combine the importance of the Hoff process in Theorem \ref{eq:hoff_converegence}, that allows us to explicitly state the PnL of the exogenous sig-trader without the need for an integral at all.

\begin{lemma} \label{thm:sig_approx}
    Let $\mathcal{Z}^{f}_{0,T} \subset K \subset C^{1-var}([0,T],\mathbb{R}^d)$ be a compact set of market factor trajectories. Then for any exogenous signature trading strategy $\xi = \phi(\hat{Z}_{0,T})$ that acts on paths in $K$ and for every $\epsilon > 0$, there exists a linear functional $\ell \in T((\R^{N+d+1})^*)$, such that
    \begin{align*}
    \sup_{Z \in K} \lVert \phi(\hat{Z}_{0,T}) - \langle \ell, \hat{\Z}_{0,T} \rangle \rVert < \epsilon.
    \end{align*}
\end{lemma}
\begin{proof}
    We can see that this result follows from the universal approximation of continuous functions on paths by linear functionals acting on the signature (Theorem \ref{eq:universal_approx}).
\end{proof}

\begin{definition}{(Trading Strategy PnL).} \label{eq:pnl_defn}
    Let $\hat{Z} = (t,X_t, f_t)$ be a market factor process with $X$ a tradable underlying process and $f$ an untradable  signal process. If $\xi \in \A^{f,\text{sig}}$ is a linear signature trading strategy such that $\xi_t^m = \langle \ell_m, \hat{\Z}_{0,t} \rangle$, for each asset $m=1,\dots,d$, then we define the PnL of the signature trading strategy between time $0$ and time $T$ as
    \begin{align} \label{eq:sig_pnl}
        V_T = \sum_{m=1}^d \int^T_0 \langle \ell_m, \hat{\Z}_{0,t} \rangle dX_t^m
    \end{align}
    where the integral is understood in the It\^{o} sense. 
\end{definition}
\noindent It is crucial to distinguish the difference between the It\^{o} integral used in this definition versus the Stratonovich integral in (2) of Example \ref{eq:example_sigs}. If that integral was in fact an It\^{o} integral then we could describe the above definition of PnL directly in terms of the add-time signature $\hat{\Z}_{0,t}$, however this is unfortunately not the case! So in order to develop a more friendly version of $V_T$, we require the following theorem involving the \textit{Hoff lead-lag process} as seen in Definition \ref{eq:hoff_defn}. \par

\begin{theorem}{(Recovery of It\^{o} Integral using the Hoff process, (Theorem 5.1, \cite{Flint2016DiscretelyProcess})).} \label{eq:hoff_converegence} 
    Let $X=(X_t)_{t \in [0,T]}$ be a stochastic process on the filtered probability space $(\Omega, \F, (\F_t)_{t \in [0,T]}, \P)$. Suppose we observe piecewise smooth streams of $X$ that are discretely sampled over a sequence of times $\{t_i\}^N_{i=0}$. Let $\hat{X}^{LL} = (\hat{X}_{t_i}^{\textup{lead}}, \hat{X}_{t_i}^{\textup{lag}})^{2N}_{i=1}$ be the associated observed \textit{Hoff lead-lag transform} of $\hat{X}$ as defined in Definition \ref{eq:hoff_defn}. Let $\phi = (\phi^1, \dots, \phi^d)$ be continuous functions acting on paths of $X$, then we have that
    \begin{align*}
        \sum_{m=1}^d \int^T_0  \phi^m(\hat{X}_t^{\textup{lag}}) d\hat{X}_t^{m, \textup{lead}} \to \int^T_0 \phi(X_t) dX_t \coloneqq & \sum_{m=1}^d \int^T_0 \phi^m(X_t) dX_t^m \\
        \text{as  } & \max_{t_i, t_{i+1}} \vert t_{i+1} - t_i \vert \to 0,
    \end{align*}
    in either probability or $L^p$-norm.
\end{theorem}
\begin{proof}
     The basic idea of the proof is that the areas between the \textit{lead} and \textit{lag} components of the Hoff lead-lag process, (captured by $\hat{X}^{LL}$), introduce a correction factor in the stochastic integral limit, consequently allowing us to recover the It\^{o}, not Stratonovich, integral.
\end{proof}
This result allows us to show that applying a function to the observed lagged path, against the observed leading path, we recover the true It\^{o} integral against the stochastic process $X$ as the mesh size goes to zero. In this sense, we see that we are able to approximate asymptotically the true It\^{o} integral as instead an integral of the observed, Hoff lead-lag stream.
\begin{remark}
    We can clearly see the analogy with Theorem \ref{eq:hoff_converegence} and our definition of PnL in Definition \ref{eq:pnl_defn}. By transforming our asset price data via the Hoff lead-lag process, now regarded as a geometric rough path, the quadratic variation of the underlying process $X$ naturally arises. The path-dependent concept of volatility has been researched extensively (\cite{Breidt1998TheVolatility, Gatheral2014VolatilityRough, Jacquier2020Path-dependentModels, Guyon2022VolatilityPath-Dependent}) and this is explicitly embedded in the second order of the signature of the Hoff process. So not only does this result allow us to capture the path-dependency of volatility, which is mostly endogenous, we can enrich this with other path-dependent market factors.
\end{remark}
\begin{figure}[ht]
\centering
  \includegraphics[width=0.9\linewidth]{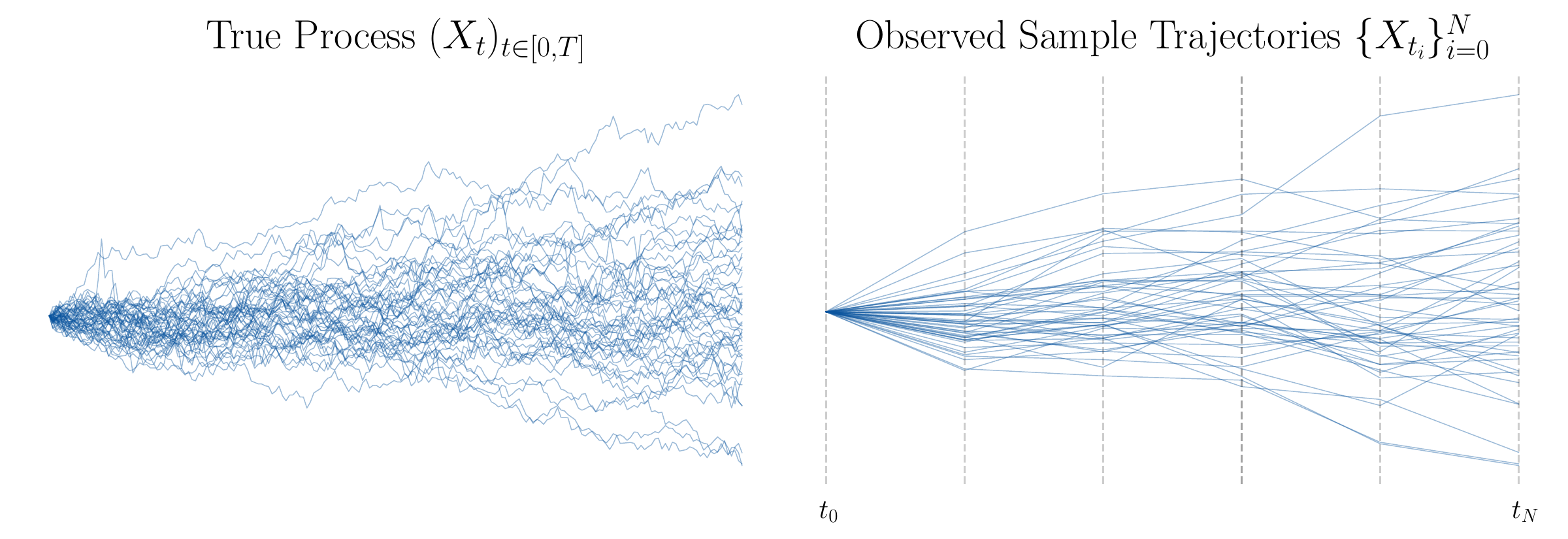}
  \caption{\centering Distinguishing between the true stochastic process $X$ and the discretely sampled observed process for which we define the Hoff lead-lag process.}
  \label{fig:true_vs_observed_v2}
\end{figure}
We can see from Figure~\ref{fig:true_vs_observed_v2} that in order for Theorem~\ref{eq:hoff_converegence} to hold, we require frequent sampling such that the distance between observations is small. A natural next step is to consider the case where our strategy $\xi$ is in fact a \textit{signature trading strategy}, i.e $\xi_t^m = \langle \ell_m , \hat{\Z}_{0,t} \rangle$ for each asset $m=1, \dots, d$. Moreover, how can we find an expression for the PnL $V_T$ in \eqref{eq:sig_pnl}, in the It\^{o} integral sense, using Theorem \ref{eq:hoff_converegence}? We will in fact extend this idea to a much more powerful result, under the presence of exogenous market signals.

\begin{restatable}{theorem}{pnlthm}{(PnL of a $d$-asset signature trading strategy under exogenous signal).} \label{thm:int_sig} 
 Let $\hat{Z} \coloneqq (t, X_t, f_t)_{t \in [0,T]}$ be the \textit{market factor process}, where $X$ is a $d$-dimensional tradable stochastic process and $f$ is a $N$-dimensional un-tradable factor process. Let $\ell_1, \dots, \ell_d \in T((\R^{N+d+1})^*)$ and define our trading strategy as $\xi_t^m =  \langle \ell_m, \hat{\Z}_{0,t}^{< \infty} \rangle$ for $m=1,\dots, d$. Then, we have that the PnL of the strategy between time $0$ and time $T$ can be represented as
\begin{align}\label{eq:pnl_leadlag} 
     V_T = \sum_{m = 1}^d \int^T_0 \langle \ell_m, \hat{\Z}_{0,t}^{< \infty} \rangle dX_t^m \approx \sum_{m = 1}^d \int^T_0 \langle \ell_m, \hat{\Z}_{0,t}^{\textup{lag}, < \infty} \rangle dX_t^{m,\textup{lead}} = \sum_{m = 1}^d \langle \ell_m \bluebf{f}(m), \hat{\Z}^{LL,<\infty}_{0,T} \rangle
 \end{align}
 where $\bluebf{f}(m): \{1, \dots, d\} \to \W(A_{2(N+d+1)})$ is a shift operator which is defined as $\bluebf{f}(m) = \pi(e^*_{m+N+d+2})$ where $\pi:T((\R^{2(N+d+1)})^*) \to \W(A_{2(N+d+1)})$ the canonical isomorphism between the dual space $T((\R^{2(N+d+1)})^*)$ and the space of all words $\W(A_{2(N+d+1)})$.
 \end{restatable}
\noindent Now, due to the linearity of expectation, we are able to represent the \textit{expected PnL} at time $T$ as the following:
\begin{align} \label{eq:exp_lead_lag_pnl}
    \Ex(V_T) & =  \sum_{m =1 }^d \langle \ell_m \bluebf{f}(m), \Ex( \hat{\Z}^{LL,<\infty}_{0,T}) \rangle.
\end{align}
where $\Ex( \hat{\Z}^{LL,<\infty}_{0,T})$ is the expected signature of the Hoff lead-lag transformation of $\hat{Z}$.
\begin{proof}
    Given in Appendix \ref{sec:proof_of_pnl_thm}.
\end{proof}
\begin{figure}[ht]
\centering
  \includegraphics[width=\linewidth]{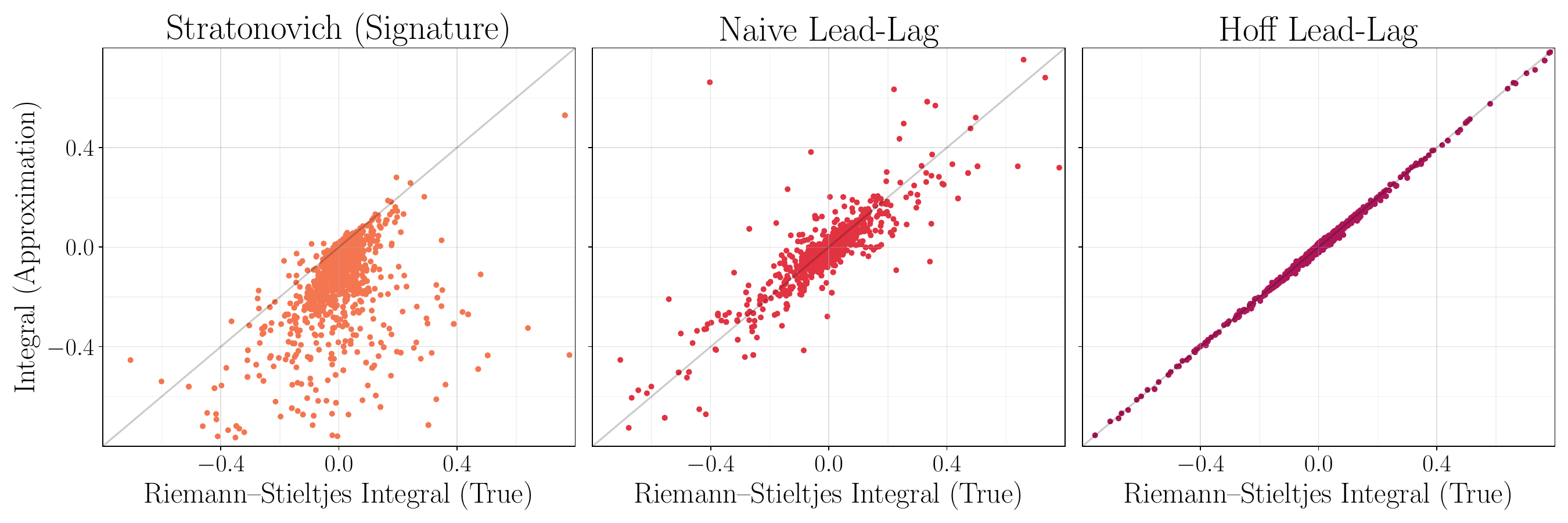}
  \caption{\centering A comparison between different approximations of the true It\^{o} PnL.}
  \label{fig:hoff_lead_lag_integral}
\end{figure}
\begin{remark}
    This result allows us to express any integral of a linear functional on the signature of a time-augmented market factor path as a newly defined linear functional on the time-augmented lead-lag market factor path instead. We must note that the authors in (\cite{Lyons2019Non-parametricDerivatives}, Lemma 3.11), construct a version in the specific case when $d=1$, $f=\emptyset$ and so $\bluebf{f}(m) = \bluebf{4}$. Here, we extend this case to allow for any exogenous market information, as well as the multi dimensionsal case for $d>1$. 
\end{remark}
\begin{remark}
    Note that $\bluebf{f}(m): \{1,\dots,d\} \to \W(A_{2(N+d+1)})$ is simply just the shift operator that allows any multi-index $(j_1, \dots, j_n) \in \left\{1,\dots,d \right\}^n$ for the original signature terms to be transformed into a multi-index for the signature terms of the $2(N+d+1)$-dimensional time-augmented lead-lag process. \par
\end{remark}
\begin{example} \label{ex:word_example}
    Define the $4$-dimensional market factor process $\hat{Z}\coloneqq (t, X^1, X^2, f^1)$ for two assets and one corresponding factor. Consider the case where our signature trading strategy is truncated at level $M=2$, then the truncated expected signature, $\Ex( \hat{\Z}^{LL,\leq M}_{0,T})$ will have 
    $$
    \vert \W_{d+N+1}^M \vert = \vert \W_{4}^{2} \vert = \sum^2_{k=0} 4^k = 21
    $$
    terms. Notably, the 21 \textit{words} associated with these signature terms $\W_{4}^{2}$ are defined as:
    $$
    \W_{4}^{2} = \left\{ \bluebf{\emptyset}, \bluebf{0}, \bluebf{1}, \bluebf{2}, \bluebf{3}, \bluebf{00}, \bluebf{01}, \bluebf{02}, \bluebf{03}, \bluebf{10}, \bluebf{11}, \bluebf{12}, \bluebf{13}, \bluebf{20}, \bluebf{21}, \bluebf{22}, \bluebf{23}, \bluebf{30}, \bluebf{31}, \bluebf{32}, \bluebf{33} \right\}
    $$
    and so any linear functional associated with the truncated expected signature, $\Ex( \hat{\Z}^{LL,\leq M}_{0,T})$, will have at most 21 non-zero terms. We can see that if we concatenate a linear functional $\ell$ with the letter $\bluebf{f}(m)$, then, it must be applied to the truncated signature $\hat{\Z}^{LL,\leq M+1}_{0,T}$ and the words associated with $\ell \bluebf{f}(m)$ will be $\bluebf{wf}(m)$ where $\bluebf{w} \in \W_{4}^{2}$. Throughout, we will refer to terms of the truncated signature via \textit{words} such as $\bluebf{w}, \bluebf{v} \in W_{d+N+1}^M$.
\end{example}

\begin{lemma}{(Variance of the PnL of a signature trading strategy under exogenous signal).} \label{thm:pnl_var}
    Let $X$ be a $d$-dimensional tradable stochastic process and let $f$ be a $N$-dimensional un-tradable factor process. Define $\hat{Z}_t \coloneqq (t, X_t, f_t)$ as the \textit{market factor process} and $\ell_1, \dots, \ell_d \in T((\R^{N+d+1})^*)$ and define our trading strategy as $ \langle \ell_m, \hat{\Z}_{0,s}^{< \infty} \rangle$ for $m=1,\dots, d$. Let the expected PnL of the strategy between time $0$ and time $T$ be defined as in (\ref{eq:exp_lead_lag_pnl}).
    Then the Variance of the PnL at time $T$ is defined as
    \begin{align*} 
        \textup{Var}(V_T) =\sum_{m =1}^d \sum_{n=1}^d \langle \ell_m \bluebf{f}(m) \shuffle \ell_n \bluebf{f}(n), \Ex( \hat{\Z}^{LL}_{0,T}) \rangle -  \langle \ell_m \bluebf{f}(m), \Ex( \hat{\Z}^{LL}_{0,T}) \rangle \langle \ell_n \bluebf{f}(n), \Ex( \hat{\Z}^{LL,}_{0,T}) \rangle.
    \end{align*}
    where $\shuffle$ is the shuffle product defined in Definition \ref{defn:shuffle_product}.
\end{lemma}
\begin{proof}
    This results follows simply from Theorem \ref{thm:int_sig} and the fact that variance of a random variable is defined as $\textup{Var}(V_T) = \Ex (V_T^2) - (\Ex(V_T))^2$, where $\Ex(V_T)$ is defined in (\ref{eq:exp_lead_lag_pnl}) and
    \begin{align}
    \label{eq:exp_PnL2} \Ex (V_T^2) = & \sum_{m =1}^d \sum_{n =1}^d \langle \ell_m \bluebf{f}(m) \shuffle \ell_n \bluebf{f}(n), \Ex( \hat{\Z}^{LL}_{0,T}) \rangle \\ 
    \label{eq:exp_PnL3} \Ex(V_T)^2 = & \sum_{m =1}^d \sum_{n =1}^d \langle \ell_m \bluebf{f}(m), \Ex( \hat{\Z}^{LL}_{0,T}) \rangle \langle \ell_n \bluebf{f}(n), \Ex( \hat{\Z}^{LL}_{0,T}) \rangle.
    \end{align}
\end{proof}
\section{Main Result
} \label{sec:original}
In this section we derive our main result, which is an explicit and concise expression for the optimal signature trading strategy in the presence of exogenous market signals considering a pathwise version of the classical mean-variance criterion. Using the trading strategy expected PnL and variance of PnL as defined in Chapter \ref{sec:market_model}, we show that the path-dependent mean-variance optimisation problem is convex in the weights of the linear functionals $\ell_1, \dots, \ell_d$.

\begin{restatable}{theorem}{origthm}{(Optimal Signature Trading Strategy).} \label{thm:orig_solution} 
    Denote $T \in \N$ the terminal time.  Let $X$ be a $d$-dimensional tradable stochastic process and let $f$ be a $N$-dimensional un-tradable factor process. Define $\hat{Z}_t \coloneqq (t, X_t, f_t)$ as the \textit{market factor process}. Define a signature trading strategy $\xi_t^m$  through a linear functional on the signature of the market factor process, i.e. $\xi_t^m = \langle \ell_m, \hat{\Z}_{0,t} \rangle$. Then, for a given truncation level $M$, the mean-variance optimal signature trading strategy $\ell^* = (\ell_1^*,\dots,\ell_d^*),$ $\ell_m^* \in T^{(M)}((\R^{N+d+1})^*)$, satisfies
    \begin{align}\label{eq:constrained_optimisation}
    \ell_m^* :=\argmax_{\substack{\ell_m \in T^{(M)}((\R^{N+d+1})^*) \\ \text{Var}(V_T)\leq \Delta}}  \sum_{m =1 }^d  \left\langle \ell_m \bluebf{f}(m), \Ex( \hat{\Z}^{LL,< \infty}_{0,T}) \right\rangle, \quad \forall m \in \{1,\dots,d\}
\end{align}
and is given by
\begin{align*}
\langle \ell_m^*, e_{ \bluebf{w}} \rangle = \frac{((\Sigma^{\textup{sig}})^{-1} \mathbf{\mu}^{\textup{sig}} )_{\bluebf{wf}(m)}}{2\lambda}, \quad m \in \{1,\dots,d\}, \bluebf{w} \in \W^M_{N+d+1}
\end{align*}
where the variance-scaling parameter $\lambda$ is given by
    \begin{align*}
    \lambda = \frac{1}{2 \sqrt{\Delta}} \left( \sum_{m=1}^d \sum_{n=1}^d \sum_{ \bluebf{w} \in \W_{N+d+1}^M} \sum_{ \bluebf{v} \in \W_{N+d+1}^M} ((\Sigma^{\textup{sig}})^{-1} \mathbf{\mu}^{\textup{sig}})_{\bluebf{wf}(m)} ((\Sigma^{\textup{sig}})^{-1} \mu^{\textup{sig}})_{\bluebf{vf}(n)} \Sigma^{\textup{sig}}_{\bluebf{wf}(m), \bluebf{vf}(n)}   \right)^{\frac{1}{2}}.
\end{align*}
We define the “Signature PnL attribution” as the $d \cdot \vert \W_{N+d+1}^M \vert$-length vector $\mu^{\textup{sig}} = (\mu^{\textup{sig}}_1, \dots, \mu^{\textup{sig}}_d)^\top$ as
\begin{align} \label{eq:mu_sig}
    \mu^{\textup{sig}}_{\bluebf{wf}(m)} = \left\langle \bluebf{wf}(m), \Ex (\hat{\Z}^{LL, < \infty}_{0,T}) \right\rangle, \quad & \forall \bluebf{w} \in \W_{N+d+1}^M,  m \in \{1,\dots,d\} 
\end{align}

\noindent and the “Signature PnL covariances” as the $d \cdot \vert \W_{N+d+1}^M \vert \times d \cdot \vert \W_{N+d+1}^M \vert $ matrix $\Sigma^{\textup{sig}}$ as
\begin{align} \label{eq:Sigma_sig}
\Sigma^{\textup{sig}}_{\bluebf{wf}(m),\bluebf{vf}(n)} = \left\langle \bluebf{wf}(m) \shuffle \bluebf{vf}(n), \Ex (\hat{\Z}^{LL, < \infty}_{0,T}) \right\rangle - \left\langle \bluebf{wf}(m), \Ex (\hat{\Z}^{LL, < \infty}_{0,T}) \right\rangle \left\langle \bluebf{vf}(n), \Ex (\hat{\Z}^{LL, < \infty}_{0,T}) \right\rangle
\end{align}

\noindent for all $\bluebf{w},\bluebf{v} \in \W_{N+d+1}^M$ and $m, n \in \{1,\dots,d\}$.
\end{restatable}

\begin{proof}
 In order to solve the constrained optimisation problem \eqref{eq:constrained_optimisation}, we can introduce the Lagrangian
\begin{align*}
    \L (\ell_1, \dots, \ell_d, \lambda) = \Ex (V_T) - \lambda(\text{Var}(V_T)-\Delta)
\end{align*}
and we are interested in finding saddle points $\ell_1^*, \dots, \ell_d^* \in T^{(M)}((\R^{N+d+1})^*), \lambda^0 \in \R$ that satisfy $\nabla \L(\ell_1^*, \dots, \ell_d^*, \lambda) = 0$. For this purpose recall Theorem \ref{thm:int_sig}, where the expected terminal PnL is given by \eqref{eq:exp_lead_lag_pnl}.
Furthermore, we can decompose the variance of the terminal PnL, given in equations (\ref{eq:exp_PnL2}) and (\ref{eq:exp_PnL2}).
Next, we can compute for each asset $m \in \{1,\dots,d\}$ and each word $\bluebf{w} \in \W_{N+d+1}^M$, the gradients of \eqref{eq:exp_lead_lag_pnl}, \eqref{eq:exp_PnL2} and \eqref{eq:exp_PnL3} with respect to $\langle \ell_m, e_{\bluebf{w}} \rangle$
\begin{align*}
\frac{\partial \Ex (V_T)}{\partial \langle \ell_m, e_{\bluebf{w}} \rangle}  & = \langle \bluebf{wf}(m), \Ex (\hat{\Z}^{LL,<\infty}_{0,T})\rangle \\
\frac{\partial \Ex (V_T^2)}{\partial \langle \ell_m, e_{\bluebf{w}} \rangle} & = 2 \sum_{n=1}^d \sum_{\bluebf{v} \in \W_{N+d+1}^M} \langle \ell_n, e_{\bluebf{v}} \rangle \langle \bluebf{wf}(m) \shuffle \bluebf{vf}(n), \Ex (\hat{\Z}^{LL,<\infty}_{0,T}) \rangle 
\\
\frac{\partial \Ex (V_T)^2}{\partial \langle \ell_m, e_{\bluebf{w}} \rangle} & = 2 \sum_{n=1}^d \sum_{\bluebf{v} \in \W_{N+d+1}^M}  \langle \ell_n, e_{\bluebf{v}} \rangle \langle \bluebf{wf}(m), \Ex (\hat{\Z}^{LL,<\infty}_{0,T}) \rangle  \langle \bluebf{vf}(n), \Ex (\hat{\Z}^{LL,<\infty}_{0,T}) \rangle 
\end{align*}
Using these computed gradients, we can compute the gradient of the Lagrangian with respect to $\langle \ell_m, e_{\bluebf{w}} \rangle$
\begin{align*}
\frac{\partial \L (\ell_1, \dots, \ell_d, \lambda)}{\partial \langle  \ell_m, e_{\bluebf{w}} \rangle} = \langle \bluebf{wf}(m), b \rangle - 2\lambda \sum_{n=1}^d \sum_{\bluebf{v} \in \W_{N+d+1}^M} \langle \ell_n, e_{\bluebf{v}} \rangle \Sigma^{\textup{sig}}_{\bluebf{wf}(m), \bluebf{vf}(n)}
\end{align*}
where we define $\mu^{\textup{sig}} = (\mu^{\textup{sig}}_1, \dots, \mu^{\textup{sig}}_d)^\top$ as
\begin{align*}
    \mu^{\textup{sig}}_{\bluebf{wf}(m)} = \langle \bluebf{wf}(m), \Ex (\hat{\Z}^{LL, < \infty}_{0,T}) \rangle, \quad & \forall \bluebf{w} \in \W_{N+d+1}^M,  m \in \{1,\dots,d\} 
\end{align*}
and the matrix $\Sigma^{\textup{sig}}$ as
\begin{align*}
\Sigma^{\textup{sig}}_{\bluebf{wf}(m),\bluebf{vf}(n)} = \langle \bluebf{wf}(m) \shuffle \bluebf{vf}(n), \Ex (\hat{\Z}^{LL, < \infty}_{0,T}) \rangle - \langle \bluebf{wf}(m), \Ex (\hat{\Z}^{LL, < \infty}_{0,T}) \rangle \langle \bluebf{vf}(n), \Ex (\hat{\Z}^{LL, < \infty}_{0,T}) \rangle
\end{align*}
for all $\bluebf{w},\bluebf{v} \in \W_{N+d+1}^M$ and $m, n \in \{1,\dots,d\}$.

Using the first order conditions, observe that we obtain the system of linear equations
\begin{align}\label{eq:system}
\mu^{\textup{sig}}_{\bluebf{wf}(m)} = 2 \lambda \sum_{n=1}^d \sum_{ \bluebf{v} \in \W_{N+d+1}^M} \langle \ell_n, e_{ \bluebf{v}} \rangle \Sigma^{\textup{sig}}_{\bluebf{wf}(m), \bluebf{vf}(n)}, m \in \{1,\dots,d \}, \bluebf{w} \in \W_{N+d+1}^M.
\end{align}
Under truncation, we find that \eqref{eq:system} is a system of
\begin{align*}
d_M & = \vert I \vert \cdot \vert \W_{N+d+1}^M \vert \\
 & = d \sum_{k=0}^M (N+d+1)^k \\
 & =  (N+d+1)^{M+1} - 1
\end{align*}
equations and $d_M$ unknowns. Hence, assuming that $\Sigma^{\textup{sig}}$ is invertible, we can solve \eqref{eq:system} and obtain for $m \in \{1,\dots,d \}, \bluebf{w} \in \mathcal{W}^M_{N+d+1}$
\begin{align}\label{eq:solution}
\langle \ell_m^*, e_{ \bluebf{w}} \rangle = \frac{((\Sigma^{\textup{sig}})^{-1} \mu^{\textup{sig}} )_{\bluebf{wf}(m)}}{2\lambda}
\end{align}
where we assumed by complementary slackness (KKT) that $\lambda \neq 0$. We can then substitute \eqref{eq:solution} into the variance constraint to obtain two solutions for $\lambda$
\begin{align*}
\lambda_\pm = & \pm \frac{1}{2 \sqrt{\Delta}} \left( \sum_{m=1}^d \sum_{n=1}^d \sum_{ \bluebf{w} \in \W_{N+d+1}^M} \sum_{ \bluebf{v} \in \W_{N+d+1}^M} \langle \ell_m, e_{\bluebf{w}} \rangle \langle \ell_n, e_{\bluebf{v}} \rangle \Sigma^{\textup{sig}}_{\bluebf{wf}(m), \bluebf{vf}(n)}  \right)^{\frac{1}{2}} \\
= & \pm \frac{1}{2 \sqrt{\Delta}} \left( \sum_{m=1}^d \sum_{n=1}^d \sum_{ \bluebf{w} \in \W_{N+d+1}^M} \sum_{ \bluebf{v} \in \W_{N+d+1}^M} ((\Sigma^{\textup{sig}})^{-1} \mu^{\textup{sig}})_{\bluebf{wf}(m)} ((\Sigma^{\textup{sig}})^{-1} \mu^{\textup{sig}})_{\bluebf{vf}(n)} \Sigma^{\textup{sig}}_{\bluebf{wf}(m), \bluebf{vf}(n)}   \right)^{\frac{1}{2}}
\end{align*}
Hence, we obtain the solution
\begin{align*}
    \langle \ell_m^*, e_{\bluebf{w}} \rangle = \frac{(\Delta)^{\frac{1}{2}}((\Sigma^{\textup{sig}})^{-1} \mu^{\textup{sig}})_{\bluebf{wf}(m)}}{\left( \sum\limits_{m,n \in \{1,\dots,d\}} \sum\limits_{\bluebf{w},\bluebf{v} \in W_{N+d+1}^N} ((\Sigma^{\textup{sig}})^{-1} \mu^{\textup{sig}})_{\bluebf{wf}(m)} ((\Sigma^{\textup{sig}})^{-1} \mu^{\textup{sig}})_{\bluebf{vf}(n)} \Sigma^{\textup{sig}}_{\bluebf{wf}(m), \bluebf{vf}(n)} \right)^{\frac{1}{2}}}\text{ }
\end{align*}
 for $m \in \{ 1, \dots, d \}, \bluebf{w} \in \W^M_{N+d+1}$.
\end{proof}
In the remainder of this section we first give an example how the main theorem can be used and then proceed to draw a parallel to the classical case and explain how our theorem extends classical formulas to the path-dependent case.

\begin{remark}
    The matrix $\Sigma^{\textup{sig}}$ and vector $\mu^{\textup{sig}}$ are just placeholders for different terms and combinations of $\Ex (\hat{\Z}^{LL,<\infty}_{0,T})$. Hence, all we need in order to know what our optimal functional is, is the expected Hoff lead-lag signature, and this is straightforward to compute for reasonable orders of truncation and number of assets and factors. The following example aims to provide some practical intuition behind how the optimal strategy is calculated and the relation between the signature and words on the tensor algebra.
\end{remark}

\begin{example}
    Let us consider the case of when we have two assets $X = (X^1, X^2)$, one factor $f$, such that $d=2$, $N=1$. Then we define the 4-dimensional market factor process $\hat{Z}\coloneqq (t,X,f)$. \par
     Let us fix the truncation level to be $M=2$. We remark that Example~\ref{ex:word_example} is of the same form, and so the 2nd level truncated signature $ \Z_{0,t}^{\leq 2}$ has $\vert \W^2_{4} \vert = 21$ terms, namely the associated words are:
    \begin{align} \label{eq:words_ex}
    \W_{4}^{2} = \left\{ \bluebf{\emptyset}, \bluebf{0}, \bluebf{1}, \bluebf{2}, \bluebf{3}, \bluebf{00}, \bluebf{01}, \bluebf{02}, \bluebf{03}, \bluebf{10}, \bluebf{11}, \bluebf{12}, \bluebf{13}, \bluebf{20}, \bluebf{21}, \bluebf{22}, \bluebf{23}, \bluebf{30}, \bluebf{31}, \bluebf{32}, \bluebf{33} \right\}
    \end{align}
\end{example}

Hence, the optimal linear signature trading strategy will correspond to two linear functionals (one for each asset) $\ell_1, \ell_2$, each of length 21, defined as
\begin{align*}
    \xi_t^1 &= \langle \ell_1, \hat{\Z}_{0,t}^{\leq 2} \rangle \\
    \xi_t^2 &= \langle \ell_2, \hat{\Z}_{0,t}^{\leq 2} \rangle
\end{align*}

In the above theorem, we can see the vector $\mu^{\textup{sig}} = (\mu^{\textup{sig}}_1, \mu^{\textup{sig}}_2)$ will be defined as follows:
$$
\mu^{\textup{sig}} \coloneqq \left\{\Ex \left(\hat{\Z}_{\bluebf{w f}(m)}^{LL, \leq 3} \right)\right\}_{\bluebf{w} \in \W^2_{4}, m=1,2}.
$$
Hence, $\mu^{\textup{sig}}$ will be a $d \cdot \vert \W^M_{N+d+1} \vert = 2 \times 21 = 42$ length vector, containing elements of the expected lead-lag signature of order 3.  \par
Recall that the shift operator $\bluebf{f}(m)$ is defined for each asset $m=1,2$ as:
\begin{align*}
    \bluebf{f}(1) & = \bluebf{5} \\
    \bluebf{f}(2) & = \bluebf{6}.
\end{align*}
which correspond to the 5th and 6th dimensions of the lead-lag process. Hence, we see that $\mu^{\textup{sig}}_1$ contains the expected lead-lag signature terms corresponding to the index of words
$$
    I_1 \coloneqq \left\{ \bluebf{5}, \bluebf{05}, \bluebf{15}, \bluebf{25}, \bluebf{35}, \bluebf{005}, \bluebf{015}, \bluebf{025}, \bluebf{035}, \bluebf{105}, \bluebf{115}, \bluebf{125}, \bluebf{135}, \bluebf{205}, \bluebf{215}, \bluebf{225}, \bluebf{235}, \bluebf{305}, \bluebf{315}, \bluebf{325}, \bluebf{335} \right\}
$$
and $\mu^{\textup{sig}}_2$ contains the expected lead-lag signature terms corresponding to the index of words
$$
    I_2 \coloneqq \left\{ \bluebf{6}, \bluebf{06}, \bluebf{16}, \bluebf{26}, \bluebf{36}, \bluebf{006}, \bluebf{016}, \bluebf{026}, \bluebf{036}, \bluebf{106}, \bluebf{116}, \bluebf{126}, \bluebf{136}, \bluebf{206}, \bluebf{216}, \bluebf{226}, \bluebf{236}, \bluebf{306}, \bluebf{316}, \bluebf{326}, \bluebf{336} \right\},
$$
such that we have:
\begin{align*}
    \mu^{\textup{sig}} = 
    \begin{bmatrix}
\left\langle \bluebf{5}, \Ex \left(\hat{\Z}^{LL, \leq 3} \right) \right\rangle \\
\vdots \\
\left\langle \bluebf{335}, \Ex \left(\hat{\Z}^{LL, \leq 3} \right) \right\rangle \\
\left\langle \bluebf{6}, \Ex \left(\hat{\Z}^{LL, \leq 3} \right) \right\rangle \\
\vdots \\
\left\langle \bluebf{336}, \Ex \left(\hat{\Z}^{LL, \leq 3} \right) \right\rangle
\end{bmatrix}
\Bigg\} \text{ 42 elements}
\end{align*}
Intuitively, we can think of each element of $\mu^{\textup{sig}}$ as the expected \textit{attribution} that each signature term has to a given assets future returns. For example, consider the term of the signature corresponding to the word $\bluebf{3}$, i.e $\langle \bluebf{3}, \hat{\Z}^{\leq 2} \rangle$, which corresponds to the increments of the factor signal, i.e
$$
\langle \bluebf{3}, \hat{\Z}^{\leq 2}_{0,T} \rangle = \int^T_0 \circ dZ^3_t
$$
Then the expected \textit{attribution} that this has on the increment of asset 1, can be defined as
$$
\Ex \left( \int^T_0 \langle \bluebf{3}, \hat{\Z}^{\leq 2}_{0,t} \rangle dX^1_t \right) = \left \langle \bluebf{3f}(1), \Ex \left( \hat{\Z}^{LL,\leq 3}_{0,T} \right) \right\rangle = \left \langle \bluebf{35}, \Ex \left(  \hat{\Z}^{LL,\leq 3}_{0,T} \right) \right\rangle = \mu^{\textup{sig}}_{\bluebf{35}},
$$
therefore the vector $ \mu^{\textup{sig}}$ consists of expected PnL attribution for each of the 21 terms of the signature of the factor process that we observe. \par

Now, we consider the $42 \times 42$ matrix $\Sigma^{\textup{sig}}$, defined element-wise as
\begin{align*}
\Sigma^{\textup{sig}}_{\bluebf{wf}(m),\bluebf{vf}(n)} = \left\langle \bluebf{wf}(m) \shuffle \bluebf{vf}(n), \Ex (\hat{\Z}^{LL, < \infty}_{0,T}) \right\rangle - \left\langle \bluebf{wf}(m), \Ex (\hat{\Z}^{LL, < \infty}_{0,T}) \right\rangle \left\langle \bluebf{vf}(n), \Ex (\hat{\Z}^{LL, < \infty}_{0,T}) \right\rangle
\end{align*}
for all $\bluebf{w},\bluebf{v} \in \W_{4}^2$ and $m, n =1,2$. \par

Each element of this matrix represents a covariance term between the PnL \textit{attributions} that we discussed previously. For example, let us observe an arbitrary element of $\Sigma^{\textup{sig}}$. Let $w = \bluebf{01}, \bluebf{v} = \bluebf{23}, m=1, n=2$.  Then $\bluebf{wf}(m) = \bluebf{015}, \bluebf{vf}(n) = \bluebf{236}$ and the corresponding element in the matrix is given as
$$
\Sigma^{\textup{sig}}_{\bluebf{015},\bluebf{236}} = \left\langle \bluebf{015} \shuffle \bluebf{236}, \Ex (\hat{\Z}^{LL, \leq 6}_{0,T}) \right\rangle - \left\langle \bluebf{015}, \Ex (\hat{\Z}^{LL, \leq 3}_{0,T}) \right\rangle \left\langle \bluebf{236}, \Ex (\hat{\Z}^{LL, \leq 3}_{0,T}) \right\rangle
$$
where $\bluebf{015} \shuffle \bluebf{236}$ is a sum of 20 different words in $\W_4^6$. We refer the reader also to Example~\ref{ex:shuffle_ex} for another example of the shuffle product. It is evident that, while our linear functional is only applied to the second order truncated signature, we require the sixth order signature in order to compute the covariance matrix, which can cause a computational bottleneck in practice. \par

Piecing this altogether, to obtain our optimal solution $\ell^* = (\ell_1^*, \ell_2^*)$, we have
\begin{align*}
    \ell^* = (\Sigma^{\textup{sig}})^{-1} \mu^{\textup{sig}}
\end{align*}
of which we obtain a 42-dimensional vector $\ell^*$ consisting of two length 21 vectors $\ell_1^*, \ell_2^*$. We can then compute the trading strategy, for each time $t$, as
\begin{align*}
    \xi_t^1 &= \langle \ell_1^*, \hat{\Z}_{0,t}^{\leq 2} \rangle \\
    \xi_t^2 &= \langle \ell_2^*, \hat{\Z}_{0,t}^{\leq 2} \rangle.
\end{align*}
Note also how we can explicitly calculate the expected PnL and variance of the portfolio explicitly in this framework. Let $\ell$ be the 42-length vector $\ell = (\ell_1^*, \ell_2^*)$, then we have 
\begin{align*}
    \Ex(V_T) &= \ell^\top \mu^{\textup{sig}} \\
    \textup{Var}(V_T) &= \ell^\top \Sigma^{\textup{sig}} \ell.
\end{align*}

\subsection{Sig-Factor Model vs Classical Factor Model}

Since $\hat{Z}$ is embedded with market factors $f$, the expected lead-lag signature $\Ex (\hat{\Z}^{LL,<\infty}_{0,T})$ will contain a wealth of path-dependent characteristics about our assets and how they are driven by the past price trajectory and the past trajectory of the market factors. At this point, we observe this solution is analogous is to the classical framing of an optimal factor model under the mean-variance framework, seen in \eqref{eq:factor_model_predict} and \eqref{eq:mean_var_sol}. In the classical factor model framework, for $N$ factors $f = f^1, \dots, f^N$, to obtain the the $m$-th asset position at time $t$, $\pi_t^m$, we have
\begin{align} \label{eq:factor_lin_functional}
    \pi_t^m = &  \left\langle \frac{1}{\lambda} (\Sigma^{-1} B)_m , f_t \right\rangle \\
    \notag = &  \left\langle \beta_m , f_t \right\rangle
\end{align}
where $(\Sigma^{-1} B)_m$ is the $m$-th row of the $d \times N$ matrix $\Sigma^{-1} B$. Therefore $\beta_m \coloneqq \frac{1}{\lambda} (\Sigma_t^{-1} B)_m$ represents a risk-weighted transformation of the coefficients used in the prediction step. This sequence of $N$ coefficients are then applied to the factors via an inner product to obtain the position for the $m$-th asset, $\pi_t^m \in \R$. We can see just how similar this example is to the sig-factor model.
\begin{align*} 
\xi_t^m = & \left\langle \frac{1}{2\lambda} ((\Sigma^{\textup{sig}})^{-1} \mu^{\textup{sig}})_m, \hat{\Z}_{0,t} \right\rangle \\
\notag = &  \left\langle \ell_m , \hat{\Z}_{0,t} \right\rangle
\end{align*}
Here, we have that the vector $\mu^{\textup{sig}}$ consists of the expected lead-lag signature PnL terms, for all $\vert \W_{N+d+1}^M \vert$ terms in the $M$-th order truncated signature that end in the letter $\bluebf{f}(m)$, i.e..
$$
\mu^{\textup{sig}} = \left\{\Ex \left(\hat{\Z}_{\bluebf{w f}(m)}^{LL, <\infty} \right)\right\}_{\bluebf{w} \in \W^M_{N+d+1}, m\in \{1,\dots,d\}}.
$$
and so the linear functional $\ell_m$ is a risk-adjusted weighting of coefficients that are applied to the signature terms. A sensible question now to ask is - would we produce the same portfolio if the factors $f = f^1, \dots, f^N$ in \eqref{eq:factor_lin_functional} were the terms of the signature? In this case, the answer is no. The sole reason for this is due to the \textit{prediction} phase in the classical factor model which induces asymmetric errors that are compounded when applied to the covariance matrix, and subsequently observed factors at time $t$, meaning the linear functionals $\beta_m$ and $\ell_m$ would not be the same. \par

\begin{figure}[ht]
\centering
  \includegraphics[width=\linewidth]{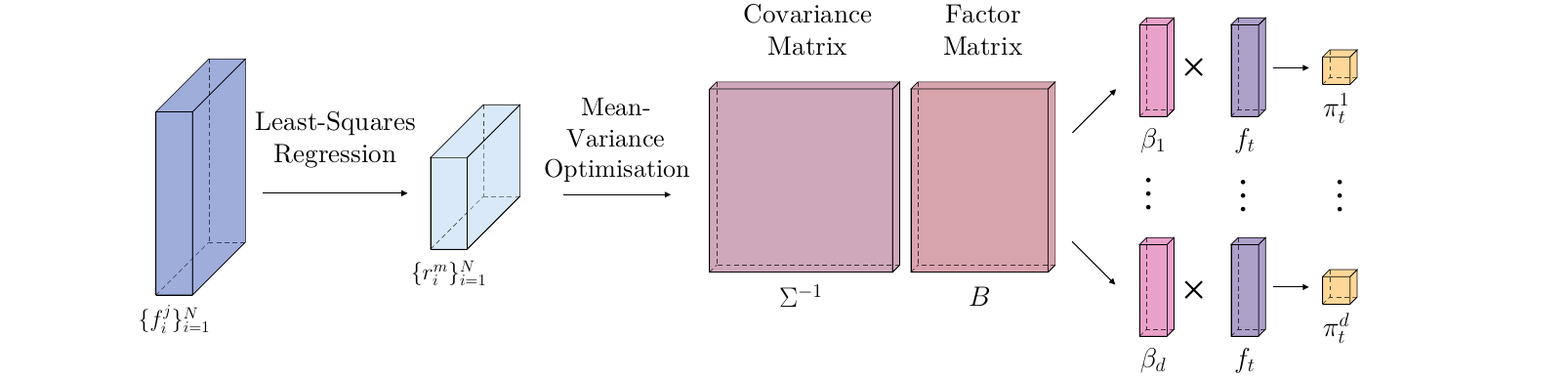}
  \caption{\centering Classical Factor Model}
  \label{fig:classic_factor_model}
\end{figure}

\begin{figure}[ht]
\centering
  \includegraphics[width=\linewidth]{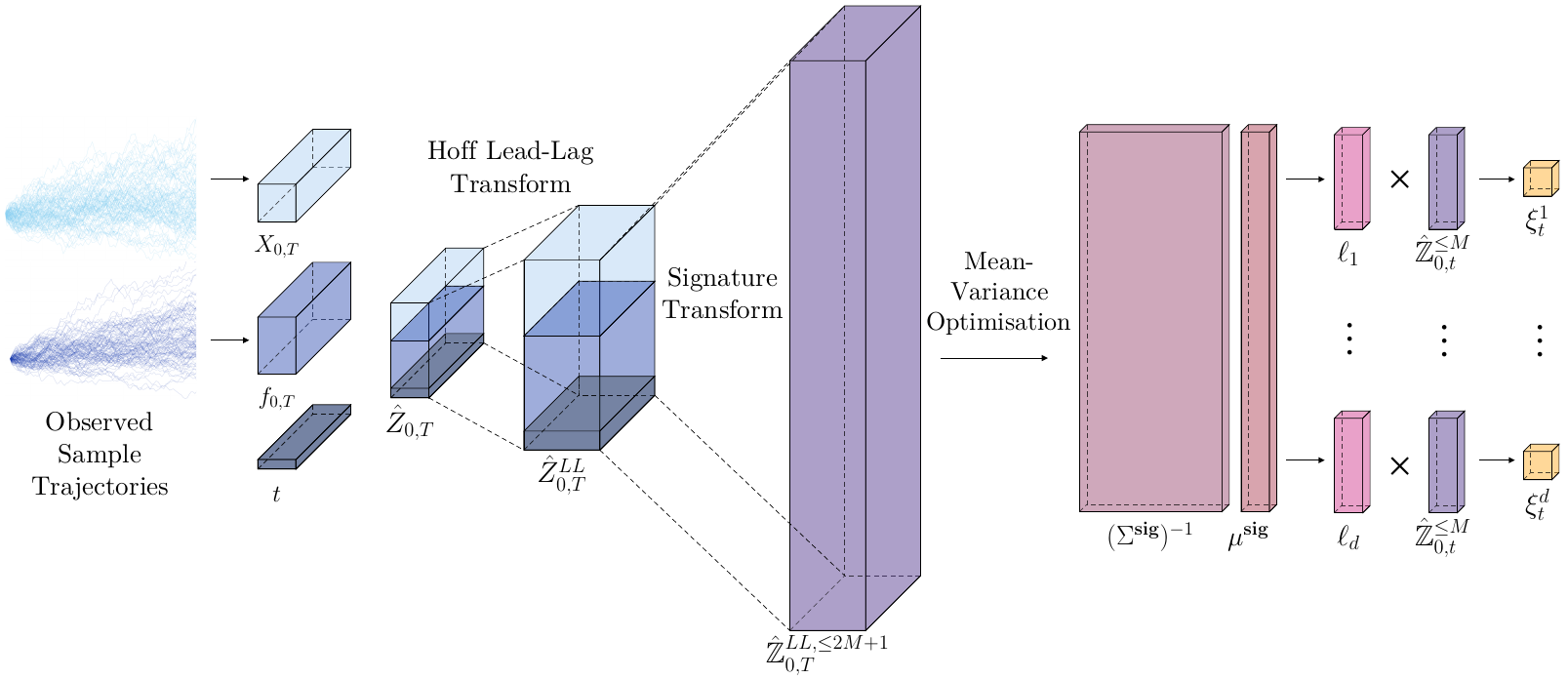}
  \caption{\centering Sig-Factor Model}
  \label{fig:sig_factor_model}
\end{figure}

Due to powerful results from rough path theory, we are able able to bypass the prediction phase by lifting and projecting our market factor process into a much higher dimensional space. Using Theorem~\ref{eq:exp_lead_lag_pnl}, we can express the expected future PnL as a linear functional on the expected Hoff lead-lag signature and then perform the optimisation \eqref{eq:constrained_optimisation} in this much higher-dimensional space, without inducing any errors originating from a least-squares regression. 

\subsection{Optimal Static Portfolio}

By design, a \textit{Signature Trading Strategy}, $\xi_t = \langle \ell, \hat{\Z}_{0,t} \rangle$, is a dynamic strategy that continuously updates its position at time $t$, depending on the value of $\hat{\Z}_{0,t}$. However, we recall that the signature is defined at each level as
$$
\hat{\Z}_{0,t} = (1, \hat{\Z}_{0,t}^1, \dots, \hat{\Z}_{0,t}^N, \dots)
$$
where the $k$-th level has $d^k$ elements. We observe that in fact the zero-th level of the signature is equal to 1, and so if we choose to only trade depending on this level of the signature, any linear functional $\ell$ applied to 1, will just return a static position for all time $t \in [0,T]$. Therefore, for a $d$-asset portfolio, we obtain
\begin{align*}
    \xi_t^1 & = \langle \ell_1, 1 \rangle = \ell_1 \in \R, \quad \forall t \in [0,T] \\
    \vdots &  \\
    \xi_t^d & = \langle \ell_d, 1 \rangle = \ell_d \in \R, \quad \forall t \in [0,T].
\end{align*}
Using Theorem~\ref{thm:orig_solution}, we have that $\ell = (\ell_1, \dots, \ell_d) \in \R^d$, where
$$
\ell = \frac{1}{2\lambda} (\Sigma^{\textup{sig}})^{-1} \mu^{\textup{sig}}.
$$
Since we are trading with respect to the zero-th order of the signature, then the only word $\bluebf{w}$ that we are interested in is $\bluebf{w} = \bluebf{\emptyset}$, therefore the number of words in our linear functional is $\vert \W_{N+d+1}^0 \vert = 1$. Hence, we can define the $d$-dimensional vector
\begin{align*}
    \mu^{\textup{sig}}_{\bluebf{wf}(m)} & = \left\langle \bluebf{wf}(m), \Ex (\hat{\Z}^{LL, < \infty}_{0,T}) \right\rangle, \quad \forall \bluebf{w} \in \W_{N+d+1}^0,  m \in \{1,\dots,d\} \\
     & = \left\langle \bluebf{f}(m), \Ex (\hat{\Z}^{LL, \leq 1}_{0,T})\right\rangle, \quad  m \in \{1,\dots,d\}.
\end{align*}
We can observe that in fact, the elements of the $d$-dimensional vector $\mu^{\textup{sig}}$ are simply the expected returns of each asset, i.e. for element corresponding to the $m$-th asset,
$$
\mu^{\textup{sig}}_m = \left\langle \bluebf{f}(m), \Ex (\hat{\Z}^{LL, \leq 1}_{0,T})\right\rangle = \Ex \left( \int^T_0 dX_t^m \right) = \Ex(X_T^m) - \Ex(X_0^m).
$$
Therefore, we have
\begin{align*}
    \mu^{\textup{sig}} = 
    \begin{bmatrix}
\Ex(X_T^1) - \Ex(X_0^1) \\
\vdots \\
\Ex(X_T^d) - \Ex(X_0^d)
\end{bmatrix}
\end{align*}
which corresponds to the expected returns vector. Likewise, we obtain the $d \times d$ matrix $\Sigma^{\textup{sig}}$ as
\begin{align*}
\Sigma^{\textup{sig}}_{m,n} = \left\langle \bluebf{f}(m) \shuffle \bluebf{f}(n), \Ex (\hat{\Z}^{LL, \leq 2}_{0,T}) \right\rangle - \left\langle \bluebf{f}(m), \Ex (\hat{\Z}^{LL,  \leq 1}_{0,T}) \right\rangle \left\langle \bluebf{f}(n), \Ex (\hat{\Z}^{LL,  \leq 1}_{0,T}) \right\rangle,
\end{align*}
which naturally corresponds to the covariance matrix of the returns for each of the $d$ assets. \par
We can see how this solution will in fact yield us the same results as the classical Markowitz portfolio, provided we use the same historical period to calculate the expected returns and covariances, of which we provide further evidence by constructing the Sig-Factor Model extension of the efficient frontier.

\subsection{Sig-Factor Model Efficient Frontier}

In Modern Portfolio Theory (MPT), first introduced in \cite{Markowitz1952PortfolioSelection}, the mean-variance optimal portfolio can be represented via the efficient frontier, which contains all portfolios that have the maximal Sharpe ratio. In the Sig-factor model, we can also obtain an efficient frontier that represents the relationship between expected returns and variance. For a given signature trading strategy $\xi_t^m = (\xi_t^1, \dots, \xi_t^d)$, where $\xi_t^m = \langle \ell_m, \hat{\Z}_{0,t} \rangle$, we can explicitly define the expected PnL and variance of our portfolio in terms of the matrix $\Sigma^{\textup{sig}}$ and vector $\mu^{\textup{sig}}$, as defined in Theorem~\ref{thm:orig_solution}, e.g we have
\begin{align*}
    \Ex(V_T) &= \ell^\top \mu^{\textup{sig}} \\
    \textup{Var}(V_T) &= \ell^\top \Sigma^{\textup{sig}} \ell.
\end{align*}
Therefore, for any linear functional $\ell$, we can observe the expected PnL and variance corresponding to it. 

\begin{figure}[ht]
\centering
  \includegraphics[width=\linewidth]{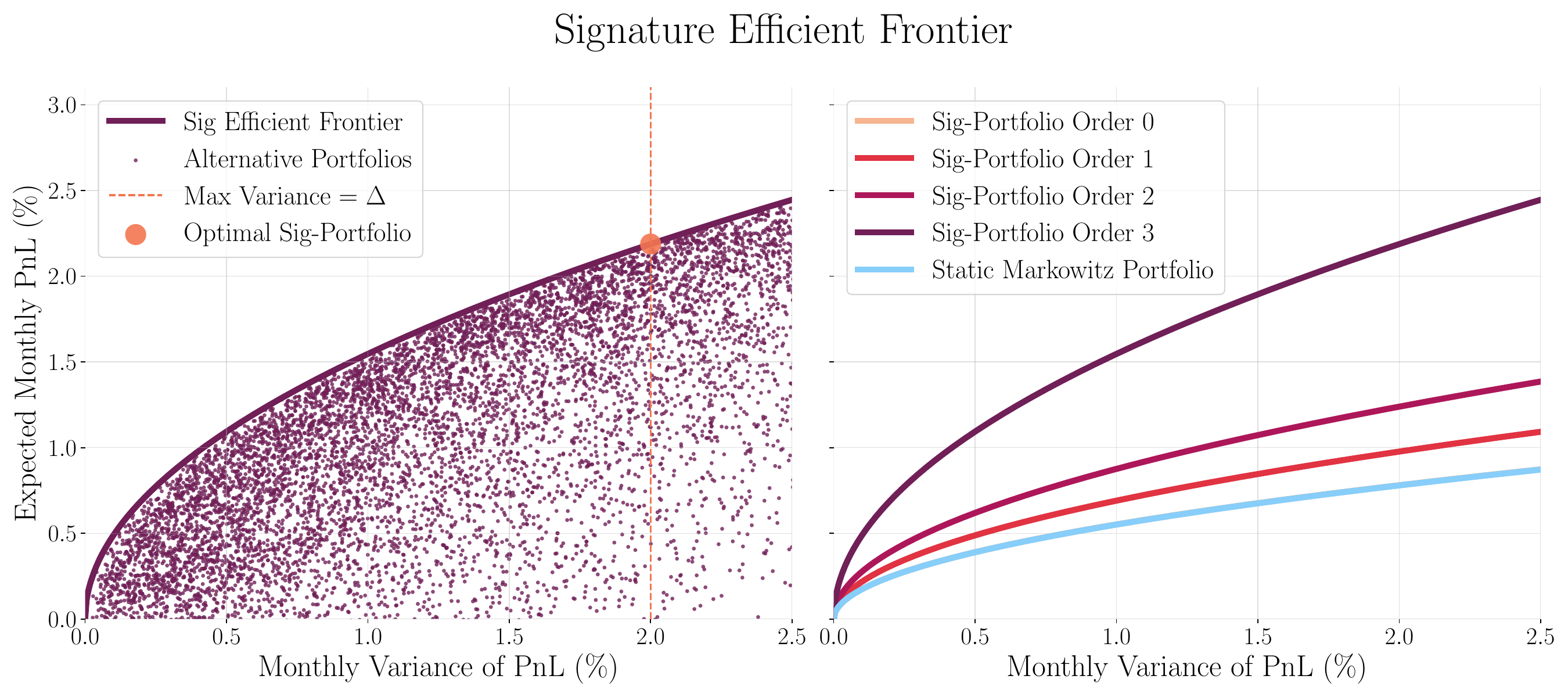}
  \caption{\centering Sig-Factor Model Efficient Frontier for 3-asset portfolio of ETFs DBC/BND/VTI.}
  \label{fig:sig_frontier}
\end{figure}

In Figure~\ref{fig:sig_frontier}, the LHS represents the level-3 Sig-trader efficient frontier curve for a portfolio consisting of the ETFs DBC/BND/VTI. This demonstrates the relationship between portfolio returns and variance, just as Markowitz had posed. The orange circle corresponds to the mean-variance optimal sig-trading strategy $\ell^*$ for a given maximum variance $\Delta$, according to the solution provided in Theorem~\ref{thm:orig_solution}. Each of the smaller dots then represent different choices of linear functional $l$, that are perturbations of the original strategy. We clearly observe that the Sig-trader portfolio maximises the risk-adjusted return (Sharpe ratio), for a given level of risk, $\Delta$. \par

On the RHS, we observe the efficient frontier curve for different orders of Sig-trader, between levels 0 to 3. Firstly, as eluded to previously, we can observe that the zero-th order Sig-trader is in fact identical to the static Markowitz portfolio, meaning that their frontier curves are not distinguishable from one another. We can also observe the improvement that is made to our expected returns, as we introduce more levels of the signature. This shouldn't be a surprise, as we know that more information about the dynamics is explained as we increase the level of truncation, which suggests that non-linear dependency structures are drivers of the future asset returns.
\vspace{-0.01cm}
\nopagebreak
\section{Implementation} \label{sec:implement}
\nopagebreak
Throughout this work, we use the package \href{https://github.com/patrick-kidger/signatory}{\texttt{signatory}} (\cite{Kidger2021Signatory:GPU}), alongside \texttt{PyTorch} for calculating and performing functionality related to tensors and the signature transform. Other packages offering signature computation include
\href{https://github.com/datasig-ac-uk/esig}{\texttt{esig}} and \href{https://github.com/bottler/iisignature}{\texttt{iisignature}}. The mean-variance Sig-Trading framework is very simple to implement and does not require heavy machinery while comparative machine learning methods may require extensive network building and hyperparameter tuning for every time the user wants to perform a new optimisation. Signature-trading is a one-model-fits-all type framework that is flexible enough to adapt to different types of signal and underlying asset process.  The algorithm to find an optimal trading strategy can be found in Algorithm \ref{alg:algo1}.

The method is also self-contained in the sense that it requires no user inputted assumptions such as expected returns or covariances and that they are inferred along with characteristics of the whole process within the algorithm itself. Once a linear functional $\ell$ is obtained from past data samples it is straightforward to unravel this into an implementable trading strategy characterised by the number of units to buy/sell at each time point $t \in [0,T]$. Since the strategy is dynamic, as new data arrives the Sig-Trader will continuously compute the signature and update their position accordingly.

Whilst Sig-Trading is data driven and does not require direct probabilistic assumptions on the underlying model, just like most frameworks, there is versatility when deploying Sig-Trading in practice as there does still remain quantities that need reliably estimating in the fitting procedure. The expected lead-lag signature could be noisy when calibrated through time and so we leave any questions of robustness (with respect to the fitting procedure) to future ongoing extensions of this project, as there remains question marks over how well defined such a solution is, especially in the context of rapidly changing regimes such as in financial markets.  In step (4) of Algorithm~\ref{alg:algo1}, we suggest taking a Monte Carlo approach by calculating the empirical expected lead-lag signature, such that $\Ex (\hat{\Z}^{LL,\leq N}) = \frac{1}{M} \sum_{i=1}^M \hat{\Z}^{LL,\leq N}_{t \in [0,T]}$. However, there remains a multitude of alternative ways to approach this for some given sample/fitting data, such as cross validation techniques, in order to ensure robustness out of sample. This paper does not directly discuss these approaches but we do point out that any traders favourite statistical estimation and training procedures can work in this setting. It may be that more recent training data is more important for fitting and this may be incorporated via a rolling fitted model - but this may lead to overfitting. Hence, in this paper we leave such statistical procedures directly to the discretion of the user and instead provide a framework in which such techniques can be ensembled. 
\vspace{0.5cm}
\begin{algorithm}[H]
    \caption{Fitting The Optimal Signature Trading Strategy} \label{alg:algo1}
    \hspace*{\algorithmicindent} \textbf{Input:               } A finite set of $M$ $d$-dimensional sample market paths $\left\{ X^i_{t \in [0,T]} \right\}_{i=1}^M$\\
    \hspace*{\algorithmicindent} \hspace*{\algorithmicindent} \hspace*{\algorithmicindent} \text{  } A finite set of $M$ $N$-dimensional sample factor (signal) paths $\left\{ f^i_{t \in [0,T]} \right\}_{i=1}^M$ \\
    \text{ } \\
    \hspace*{\algorithmicindent} \textbf{Output:         } $\ell_m \in T((\R^{N+d+1})^*) \quad \forall m \in \{ 1, \dots, d\}$: The optimal trading strategy as a \\
    \hspace*{\algorithmicindent} \hspace*{\algorithmicindent} \hspace*{\algorithmicindent} \hspace*{\algorithmicindent}
    linear functional on the signature of the time-augmented path \\
    \hspace*{\algorithmicindent} \hspace*{\algorithmicindent} \hspace*{\algorithmicindent} \hspace*{\algorithmicindent} \text{ }  \\
    \hspace*{\algorithmicindent} \textbf{Parameters:  } $ N \geq 0$: Truncation level of the signature \\
    \hspace*{\algorithmicindent} \hspace*{\algorithmicindent} \hspace*{\algorithmicindent} \hspace*{\algorithmicindent} \hspace*{\algorithmicindent} $\Delta \geq 0$: Maximum variance of PnL \\
    \begin{algorithmic}[1]
    \State Create the $(N+d+1)$-dimensional market factor process $\hat{Z} = (t,X_t,f_t)$ to obtain $\left\{ \hat{Z}_{t \in [0,T]} \right\}_{i=1}^M$.

    \text{ }
    \State Apply time-augmentation and lead-lag transformations to all market factor paths, to obtain $\left\{ \hat{Z}^{LL}_{t \in [0,T]} \right\}_{i=1}^M$.

    \text{ }
    \State Compute the truncated signature (at order $N$) of each market factor path $\left\{ \hat{\Z}^{LL,\leq N}_{t \in [0,T]} \right\}_{i=1}^M$.

    \text{ }
    \State Calculate the empirical expected signature $ \Ex (\hat{\Z}^{LL,\leq N}) = \frac{1}{M} \sum_{i=1}^M \hat{\Z}^{LL,\leq N}_{t \in [0,T]} $

    \text{ }
    \State Populate vector $\mu^{\text{sig}}$ as a subset of $\Ex (\hat{\Z}^{LL,\leq N})$ terms, using \eqref{eq:mu_sig}.
    
    \text{ }
    \State Populate matrix $\Sigma^{\text{sig}}_{\bluebf{w},\bluebf{v}}$ for each word $\bluebf{w},\bluebf{v} \in \W_{N+d+1}$, using \eqref{eq:Sigma_sig}.
        
    \text{ }
    \State Solve the system of linear equations as to obtain $\ell_m \in T((\R^{N+d+1})^*) \quad \forall m \in \{ 1, \dots, d\}$ where
    $$
    \langle \ell_m^*, e_{ \bluebf{w}} \rangle = \frac{((\Sigma^{\textup{sig}})^{-1} \mathbf{\mu}^{\textup{sig}} )_{\bluebf{wf}(m)}}{2\lambda}, \quad m \in \{1,\dots,d\}, \bluebf{w} \in \W^M_{N+d+1}
    $$
    
    \State \Return Linear functional $\ell_m \quad \forall m \in \{ 1, \dots, d\}$.
    
    \end{algorithmic}
\end{algorithm}
\vspace{0.5cm}
Likewise, while signature trading strategies do a good job of drawdown control, we do not discuss extra practicalities such as vol-scaling positions in this paper. Due to the nested nature of filtrations in this work (i.e. we are in possession of strictly increasing amounts of information through time), this may impact performance of a trading strategy differently through the life of the trade, hence it would make practical sense to smoothe out Sig-Trading positions through time - this also can reduce any bias that may arise from the exact start date of the trade. We also suggest that, especially at higher frequencies, that we could replace time-augmentation with \textit{volume-augmentation} since this is still a monotonically increasing channel in the process (ensuring the signature remains unique), but represents a new notion of time.

\begin{algorithm}[H]
    \caption{Trading The Optimal Signature Trading Strategy, at time $t$} \label{alg:algo2}
    \hspace*{\algorithmicindent} \textbf{Input:             } The previous stopped underlying asset and signal paths $X_{s \in [0,t]}, f_{s \in [0,t]}$, \\
    \hspace*{\algorithmicindent} \hspace*{\algorithmicindent} \hspace*{\algorithmicindent} \text{   } Linear functional $\ell_m \in T((\R^{N+d+1})^*) \quad \forall m \in \{ 1, \dots, d\}$. \\
    
    \hspace*{\algorithmicindent} \textbf{Output:         } $\xi^m_t \quad \forall m \in \{ 1, \dots, d\}$: The optimal trading strategy \\
    \begin{algorithmic}[1]
    \State Create the $(N+d+1)$-dimensional market factor paths and apply time-augmentation to obtain $\hat{Z}_{s \in [0,t]}$.
    
    \text{ }
    \State Compute the truncated signature of the stopped market factor path, $\hat{\Z}_{0,t}$.

    \text{ }
    \State For each asset $m \in \{ 1, \dots, d\}$, obtain the optimal strategy $\xi_t^m$ by applying an inner product of the terms of $\ell_m$ to the signature of the time-augmented market path.
    $$
    \xi_t^m= \langle \ell_m , \hat{\Z}_{0,t} \rangle
    $$
    \State \Return $\xi_t^m \quad \forall m \in \{ 1, \dots, d\}$
    
    \end{algorithmic}
\end{algorithm}

\section{Numerical Results} \label{sec:results}

Throughout this section, we aim to demonstrate and highlight some of the capabilities of the Sig-Trading framework, incorporating path-dependencies and exogenous signals.

\subsection{Synthetic Data}

\subsubsection{Pairs Trading}
First, we explore the scenario in which we have no exogenous trading signal to enrich our trading strategy but we only have access to the underlying time series of the asset process. The true advantage that data-driven methods have over classical parametric frameworks is that they can exploit the inefficiencies present in financial time series data, without specifying the explicit dynamics that they are trying to capture. In this toy example, we aim to isolate one specific aspect that is exhibited by time series data and highlight how the Sig-Trader exploits it. \textit{Pairs Trading} is one of the most famous and original active trading strategies deployed by investors that focuses on trading the joint behaviour between two assets. The sentiment is that while both assets have their own dynamics, the difference (spread) between the prices of the two assets should hold some predictability on future (co-)movements. Classical literature focuses on modelling this relationship using a mean-reverting process such as an Ornstein-Uhlenbeck process. The strategy should then contain a \textit{buy signal} when the spread falls below some threshold and a \textit{sell signal} when the spread is above the threshold, in the anticipation that this spread should converge back to the threshold. More on mean-reversion strategies can be found in \cite{Alexander2002TheStrategies, Vidyamurthy2004PairsAnalysis, Mudchanatongsuk2008OptimalApproach, Cartea2015AlgorithmicAssets, Leung2015OptimalExit}.  \par
In this experiment we take two assets such that
\begin{align*}
    dX_t & =  \sigma^X dW^X_t \\
    dY_t & =  \kappa (X_t - Y_t) dt + \sigma^Y dW^Y_t
\end{align*}
Where $X$ is a standard arithmetic Brownian motion with zero drift and volatility $\sigma^X$. $Y$ however is modelled as a mean-reverting process where its drift is proportional to the spread between $X$ and $Y$. We can clearly see in this toy example that the only exploitable \textit{alpha} within the framework is the temporal dependence through mean-reversion in $Y$ since there is no long term drift in $X$ or $Y$. \par
\begin{figure}[ht]
\centering
  \includegraphics[width=\linewidth]{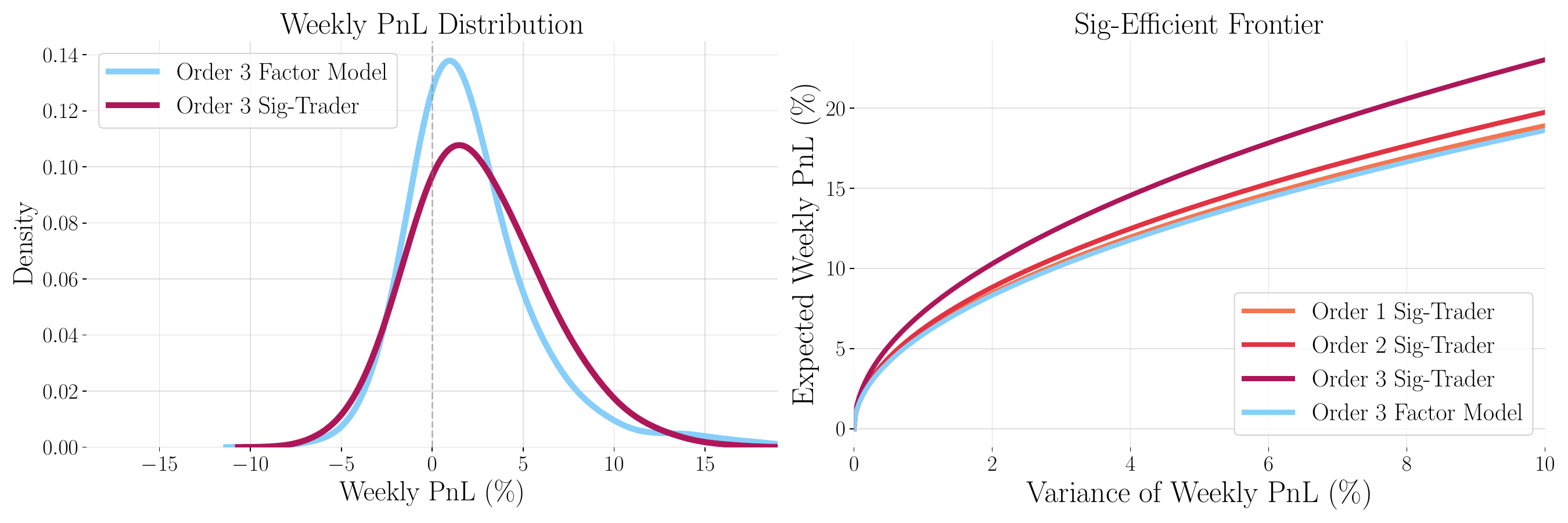}
  \caption{\centering Sharpe ratio distribution (LHS) and Sig-Trader strategy as a function of the spread between the two assets (RHS).}
  \label{fig:sim_pair_trade}
\end{figure}
 If a trader was to deploy a static buy and hold strategy here, the PnL would be zero on average, however, for higher order Sig-Traders, it is possible to exploit the mean-reversion dynamics. In fact, this example demonstrates how the above-mentioned \textit{alpha} is self-contained within the 1st level of the signature (which is the increment of the path, or the ‘drift’) and so orders of $N \geq 2$ do not contain any more predictive power on excess returns than that of $N=1$. However,  on the right hand side (RHS) of Figure~\ref{fig:sim_pair_trade}, the signature efficient frontier illustrates that when trading with a weekly look-forward horizon, the ratio of return to variance of weekly PnL is greater as we increase the order of the Sig-Trader. This is due to the higher levels of the signature capturing non-linear path-dependencies that can help reduce variance in the weekly PnL distribution. We relate back to Figure~\ref{fig:dynamic_strategy} to demonstrate that in fact higher order Sig-Traders are able to construct mean-reverting strategies that naturally limit drawdowns, which are an inherently path-dependent characteristic. \par

 In this synthetic example, we compare the Sig-Trader strategy to that of the original factor model. The generic factor model is set up as a supervised linear regression on future returns, as a function of the current signature,
\begin{align*} 
    \mu_{t+1} = \Ex [ r_{t+1} \vert \F_t ] = B \mathcal{S}_t + \varepsilon_{t+1}
\end{align*}
where $r$ is the $2$-dimensional asset returns, $B$ is a $2 \times N$ matrix of factor coefficients, $\mathcal{S}_t$ is a vector of the signature values of the  and $\varepsilon$ is a vector of the $2$ assets’ (unexplained) residuals returns. We can see that this model uses the same input as the Sig-Trader (the signature), with the same tools (applying a linear function). \par
\begin{figure}[ht]
  \centering
    \includegraphics[width=\linewidth]{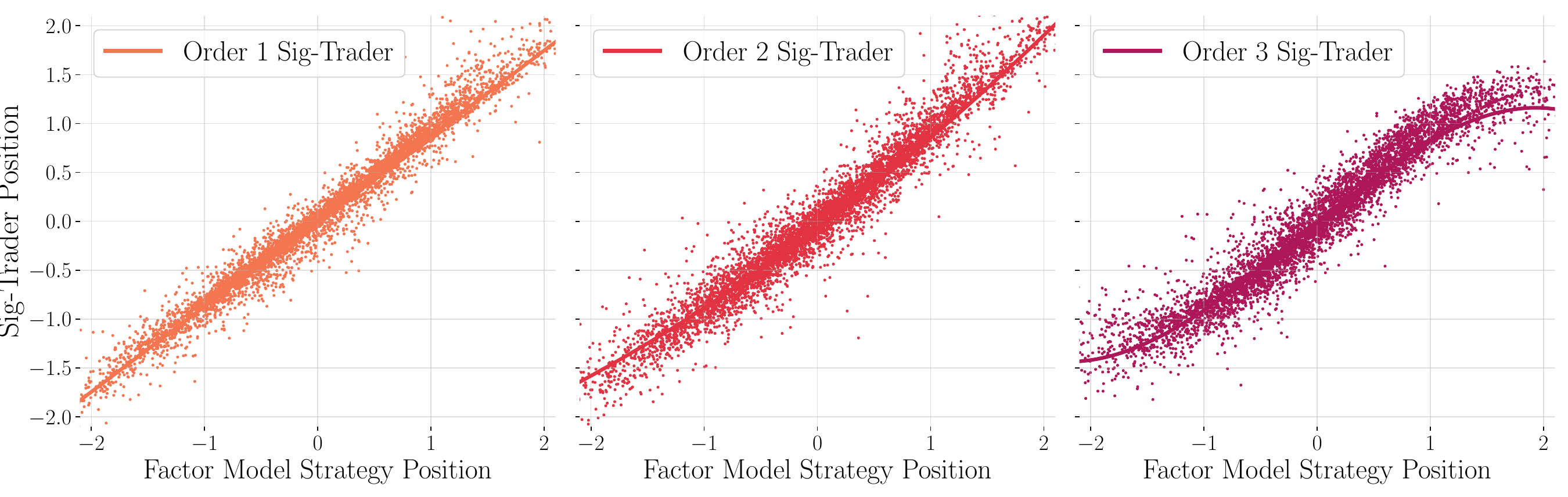}
    \caption{\centering Comparison between the Markowitz Order 3 factor model positions and the corresponding Sig Trader positions, for orders 1,2,3.}
    \label{fig:sim_pair_positions}
  \end{figure}
The key distinguishment is that the factor model is optimal for maximising the daily mean-variance PnL profile, which ignores any long term, pathwise dynamics of the strategy itself. Figure~\ref{fig:sim_pair_positions} demonstrates the difference in nature between the Sig-Trader and the factor model. In the order 1 case, as described above, the exploitable alpha in terms of returns are captured by the 1st level of the signature and so we see a similar position profile for the Sig-Trader and the factor model. However, when the factor model tends to go more short (or long), the higher order Sig-Traders, e.g order 3, tends to reduce its position since this will be more beneficial to minimising the weekly variance of PnL. We can think of this behaviur as being similar to applying a sigmoid function to your position in search of robustness, or applying a stop-loss to avoid becoming too leveraged - which are common practices. By incorporating path-dependencies in the strategy, we have access to a more robust and intutive extension to classic factor models.

\subsubsection{Incorporating Exogenous Signal}

We now consider the case when we are in possession of an exogenous signal that can be used to inform our trading decisions, rather than solely relying on raw asset time series data. Suppose the asset price process $X$ is driven by some function of the signal $\phi(t,f_{0,t})$, e.g
\begin{align*}
  dX_t = \phi(t, f_{0,t}) dt + \sigma^{X} dW_t^{X}.
\end{align*}
Since $\phi(t, f_{0,t})$ is a function of the past time series of $f$, it may be difficult for a trader to directly model the dynamics of this system directly if they do not know the explicit form of $\phi$. In practice, a trader may use a Kalman (or alternative) filter to capture the impact of a noisy signal; in fact, the authors in \cite{Cohen2023NowcastingMethods} prove how the Kalman filter can be equivalently written as a linear regression on the signature. In this example, we propose the following system for the signal process $f$ and its consequent impact on the underlying asset $X$,
\begin{align*}
      df_t & = -\kappa f_t + \sigma^{f} dW_t^{f} \\
      Z_t & = \int^t_0 K (t-s) df_s \\
      dX_t & = Z_t dt + \sigma^{X} dW_t^{X}.
\end{align*}
We let the (observable) signal $f$ be a generic mean-reverting OU process with zero mean, while it has a path-dependent causal impact on some latent process $Z$ via a time-dependent kernel $K$ that we do not observe. This kernel could simply be a stochastic filter, for example if the kernel is exponential, we recover an exponentially weighted average of previous increments in the signal. In this example, we take the kernel to be $K(t,s) = \exp\{-\alpha(t-s)\}$, which can be understood to be a decaying impact of the signal on the process $Z$, i.e more recent observations of $f$ have a greater impact on the instantaneous drift. The asset price process $X$ then behaves like an arithmetic Brownian motion with long term zero drift, but short term temporal structure. \par
\begin{figure}[ht]
  \centering
    \includegraphics[width=\linewidth]{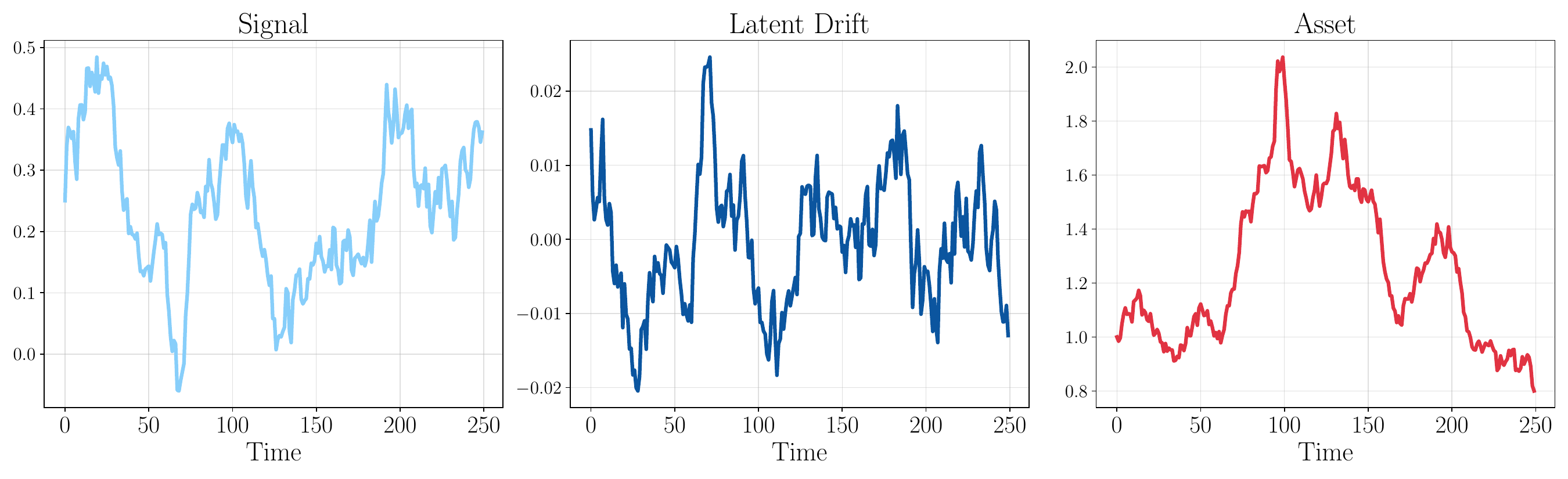}
    \caption{\centering An Example of Signal/Drift/Asset Trajectories.}
    \label{fig:synthetic_signal_example}
  \end{figure}
The underlying asset process $X$ does not have a long term drift, so we would expect a static strategy to have no long-term expected return, as seen in the left hand plot of Figure~\ref{fig:synthetic_signal_sharpes}. In fact, we notice that any excess return is marginal when the Sig-Trader trades endogenously without access to the signal, which can be seen in the light blue distributions in Figure~\ref{fig:synthetic_signal_sharpes}. The red distributions correspond to the simple factor model that takes the value of the signal at time $t$ and predicts the future (one-step) return, i.e
\begin{align} \label{eq:signal_factor}
    \mu_{t+1} = \Ex [r_{t+1} \vert \F_t ] = \beta f_t + \varepsilon_{t+1}
\end{align}
where the position is then scaled according to the expected return $mu_{t+1}$. The dark blue Sharpe ratio distributions correspond to the Sig-Trader for different orders of truncation. Clearly, we see that for higher orders of truncation, the Sharpe ratio improves as the Sig-Trader can better approximate the non-linear relationship between the signal and the underlying asset process.
\begin{figure}[ht]
\centering
  \includegraphics[width=\linewidth]{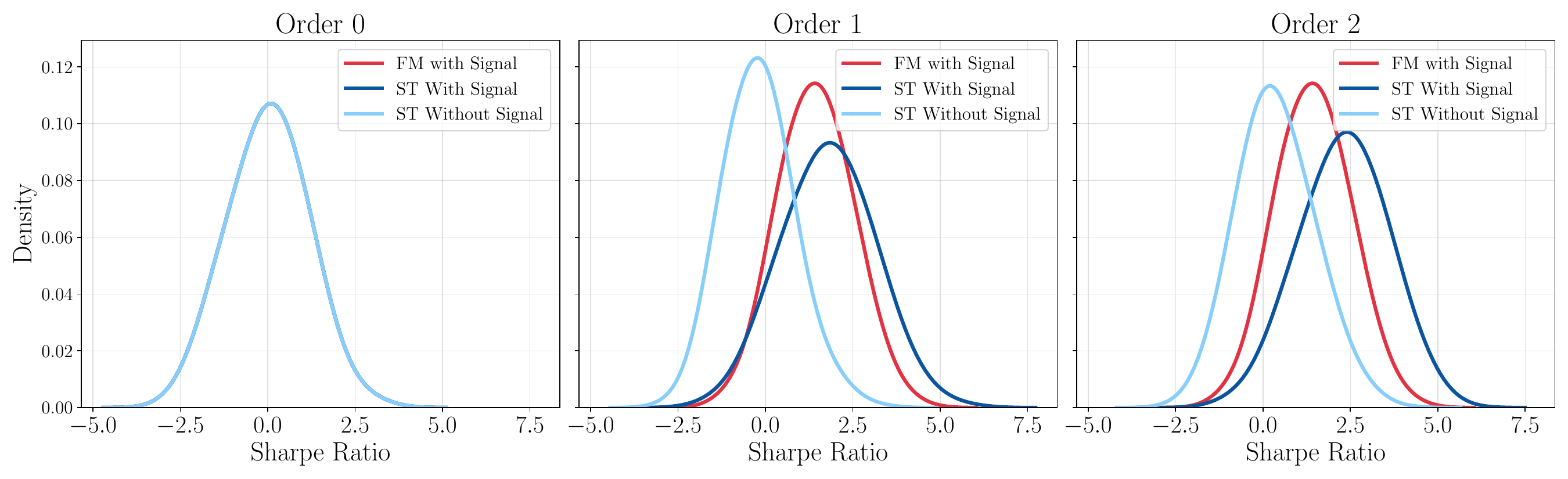}
  \caption{\centering Sharpe ratio distributions for the Sig-Trader of truncation orders 0,1,2, compared to the factor model described in \eqref{eq:signal_factor}.}
  \label{fig:synthetic_signal_sharpes}
\end{figure}
In this system there are several layers of noise and complexities to sift through, including the path-dependent impact of the signal, as well as noise from the signal itself $\sigma^f$ and exogenous noise of the underlying asset, $\sigma^X$. In this simple system, this volatility of such randomness is constant through time, but in practice this is likely not the case and this is the type of scenario when the Sig-Trader can outperform the classic predict-then-optimise frameworks. The left hand side of Figure~\ref{fig:synthetic_signal_noise_ratio} illustrates how the Sig-Trader can transform a signal of varying strengths, into a position that is able to produce strong risk-adjusted returns. 
\begin{figure}[H]
\centering
  \includegraphics[width=\linewidth]{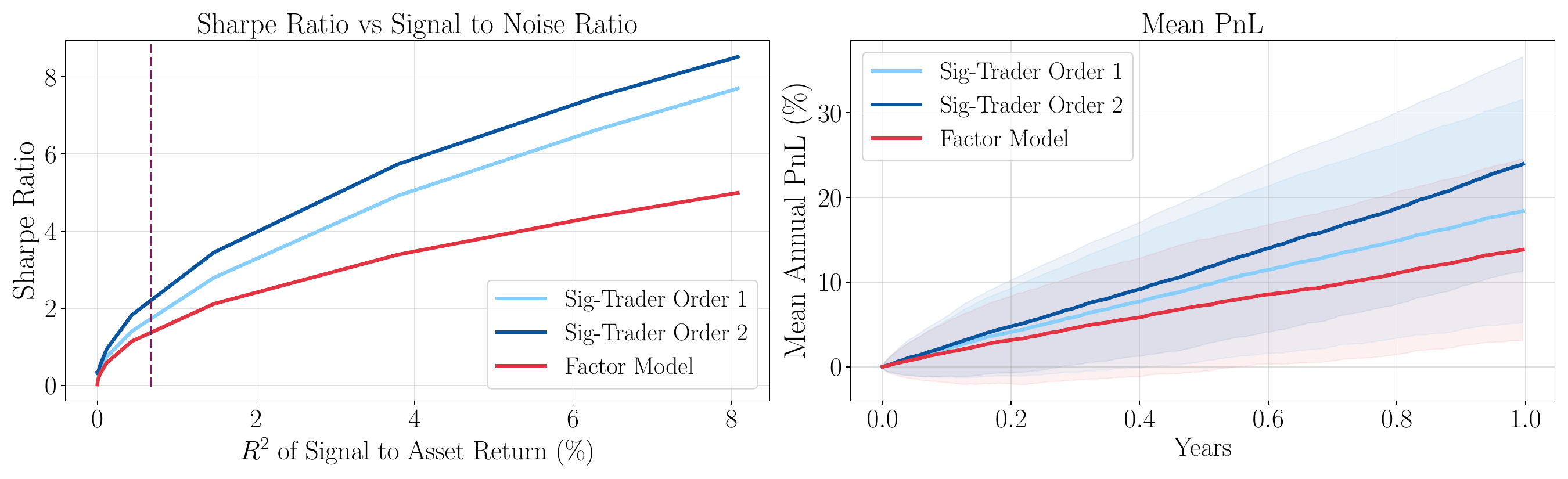}
  \caption{\centering Sharpe ratio against the signal to noise ratio (LHS). The mean PnL and variance of PnL through time (RHS).}
  \label{fig:synthetic_signal_noise_ratio}
\end{figure}

\subsection{Learning Momentum as a Sig-Trading Strategy}

As highlighted previously, the class of signature trading strategies, contains most common systematic trading strategies that are functions of the past market time-series. One of the most common and established forms of systematic trading strategy is trend following (and its multiple variants), which generally applies some function (i.e a filter) to the past market time series, to determine the strength and direction of the underlying asset trend. The trader then trades in this direction, adjusting their position according to the strength (alternatively, they may we want to trade the opposite direction, which would constitute a mean-reversion strategy). For a high-level overview of the design and mechanics of momentum strategies, we refer the reader to \cite{RichardJMartin2012MomentumMe, RichardJMartin2012Non-linearStrategies}, which give a large discussion on both linear and non-linear momentum. For practical applications, see \cite{Moskowitz2012TimeMomentum, Asness2013ValueEverywhere, Baltas2013MomentumFunds, Barroso2015MomentumMoments, Lemperiere2014RiskReturns}. \par

Within the broader class of momentum strategies live several variants, and this is dependent on your choice of function that characterises the trend. Variants (and combinations thereof) include RSI indicators, Bollinger bands, MACD (moving average convergence divergence) or any other type of moving average crossover. The key observation is that all of the previous mentioned variants constitute alternative functions (filters) of the past time series, therefore a natural question to ask is what are the best (or in fact optimal) types of function or filter to apply to the past market time series, in order to capture the dynamics of the underlying asset? Since momentum strategies are contained within the space of signature trading strategies, we are able to find the linear functional corresponding to a given momentum strategy. To make this more precise, we focus on capturing the characteristics of a MACD momentum strategy. We recall that MACD$(t_1, t_2)$ is the difference between the (exponentially weighted) $t_1$-moving average (the fast signal) and the $t_2$-moving average (the slow signal), designed to indicate strength of trend. The general setup might look as follows:

\begin{tikzpicture}
  \centering
  \hspace{-0.5cm}
  \draw[rounded corners=5pt] (0, 0) rectangle (4.5, 1);
  \node at (2.25, 0.5) {\parbox{4cm}{\centering Time Series Paths, $X_{0,t}$}};
  
  \draw[->] (4.5, 0.5) -- (5.3, 0.5);
  \node at (4.9, 0.9) {$\phi$};
  
  \draw[rounded corners=5pt] (5.3, 0) rectangle (7.5, 1);
  \node at (6.4, 0.5) {\parbox{3cm}{\centering MACD}};

  \draw[->] (7.5, 0.5) -- (8.3, 0.5);
  \node at (7.9, 0.9) {$L$};

  \draw[rounded corners=5pt] (8.3, 0) rectangle (11.3, 1);
  \node at (9.8, 0.5) {\parbox{3cm}{\centering Future Returns}};

  \draw[->] (11.3, 0.5) -- (12.1, 0.5);
  \node at (11.7, 0.9) {$\sigma$};

  \draw[rounded corners=5pt] (12.1, 0) rectangle (16.1, 1);
  \node at (14.1, 0.5) {\parbox{4cm}{\centering Strategy Position, $\xi_t$}};
\end{tikzpicture}

This is a very general framework that takes a given filter (i.e the MACD), uses it to predict future returns, and then applies some normalisation (i.e a sigmoid function) in order to retrieve a final strategy position. We note that a momentum strategy should naturally result in a positive regression coefficient of $L$ since a positive MACD signal indicates that the asset is ‘trending’ (otherwise, if the coefficient was negative, this would imply mean-reversion). We can think of this whole framework as being one continuous function of the past path, i.e
\begin{align*}
    \xi_t = \varphi(t, X_{0,t}) = \sigma(L(\phi(t, X_{0,t}))).
\end{align*}
The goal is therefore to demonstrate that the function $\varphi$ can be captured via a signature trading strategy such that $\varphi(t, X_{0,t}) = \langle \ell, \hat{\bX}_{0,t} \rangle$. \par

Figure~\ref{fig:MACD_graph} displays the corresponding MACD filter, which we denote $\phi$, while Figure~\ref{fig:learning_momentum} shows the improvement of learnt linear functionals $\ell$, as the order of truncation of the signature increases. We notice that the order 3 signature is able to approximate the function $\varphi$ with almost $90\%$ $R^2$ accuracy, even though both the filter $\phi$ and the normalisation function $\sigma$ are highly non-linear.

Given that we are able to approximate the momentum trading strategy via a Sig-Trading strategy, we can compare this strategy to the mean-variance optimal order 3 strategy. Using the learnt linear functional $\ell$, we are able to construct the efficient frontier via the expected terminal PnL and variance of terminal PnL,
\begin{align*}
    \Ex(V_T) &= \ell^\top \mu^{\textup{sig}} \\
    \textup{Var}(V_T) &= \ell^\top \Sigma^{\textup{sig}} \ell.
\end{align*}
Figure~\ref{fig:MACD_EF_convex} (LHS) demonstrates that the optimal order 3 Sig-Trader has a much better risk-return profile over the chosen trade horizon of 20 days (one month), than that of the learnt momentum trader. It is worth noting however that the Sig-Trader in this example has many similar characteristics to that of the momentum trader, however the optimal Sig-Trader was able to capture further
\begin{figure}[ht]
\centering
  \includegraphics[width=\linewidth]{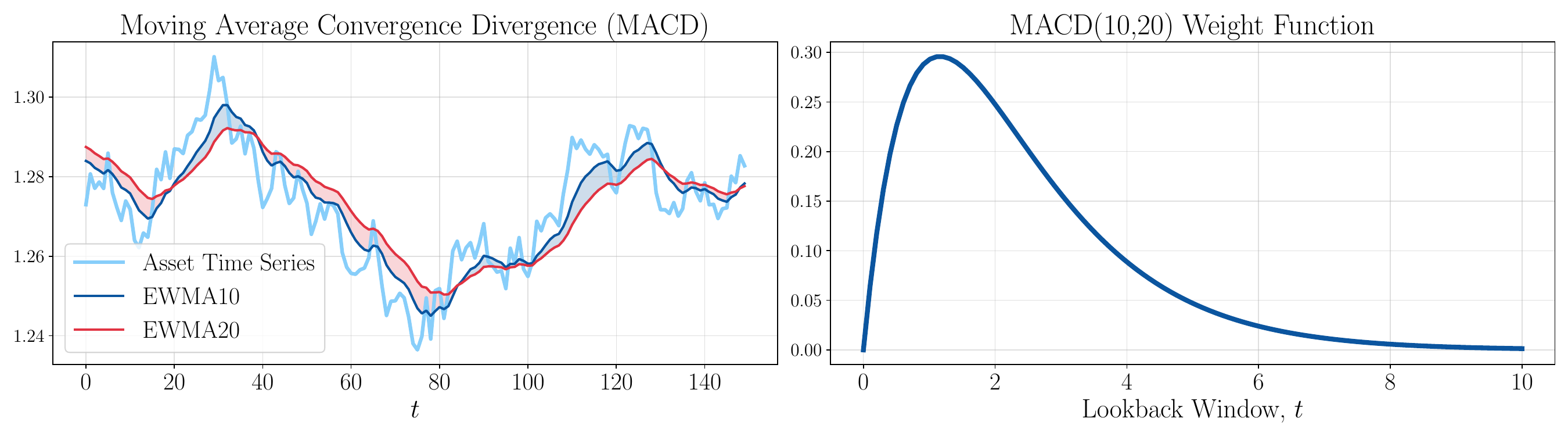}
  \caption{\centering An example of the corresponding slow and fast moving averages of a MACD(10,20) signal on the TLT ETF during 2006 (LHS). The associated weight function $\phi(t-s)$ with the MACD(10,20) filter (RHS).}
  \label{fig:MACD_graph}
\end{figure}
\begin{figure}[ht]
\centering
  \includegraphics[width=\linewidth]{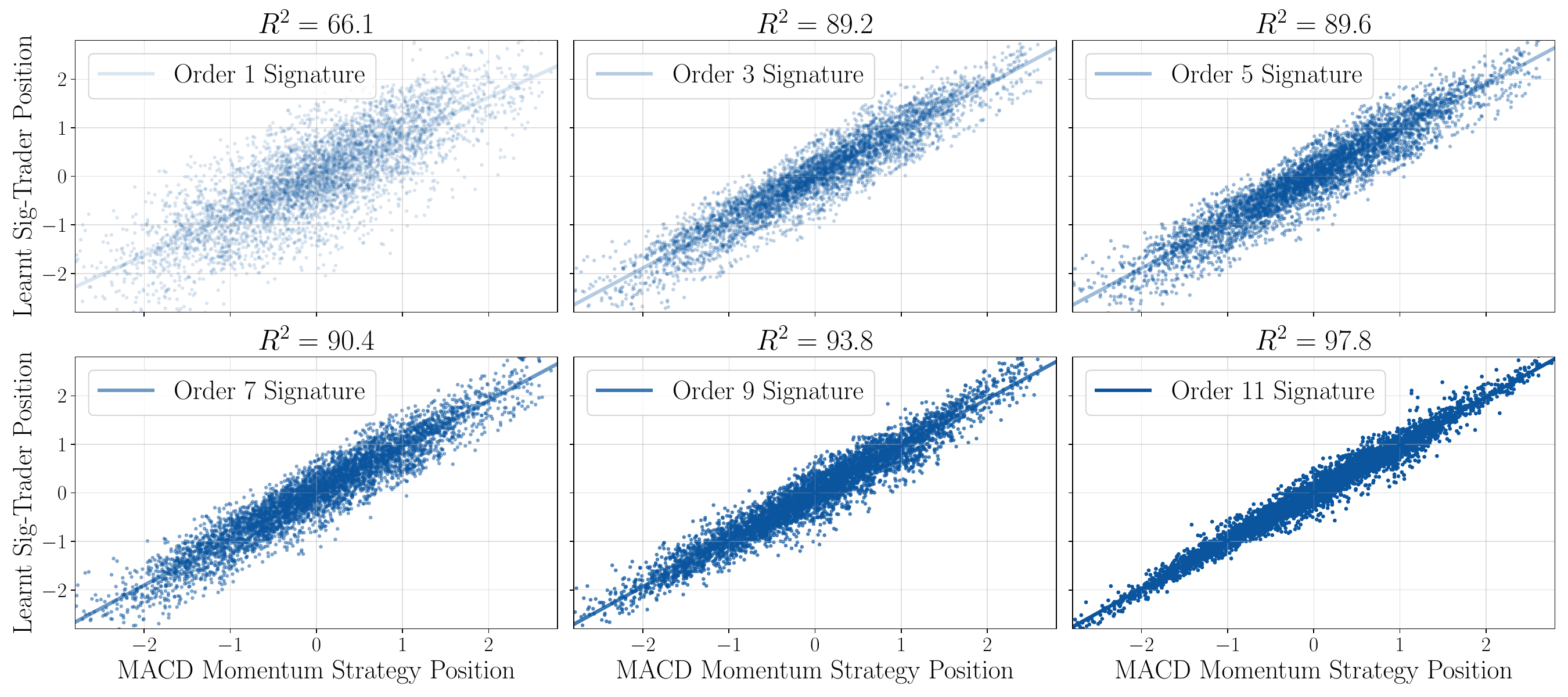}
  \caption{\centering The learnt MACD momentum strategy as a linear functional on the signature for different orders of truncation.}
  \label{fig:learning_momentum}
\end{figure}
\begin{figure}[ht]
\centering
  \includegraphics[width=\linewidth]{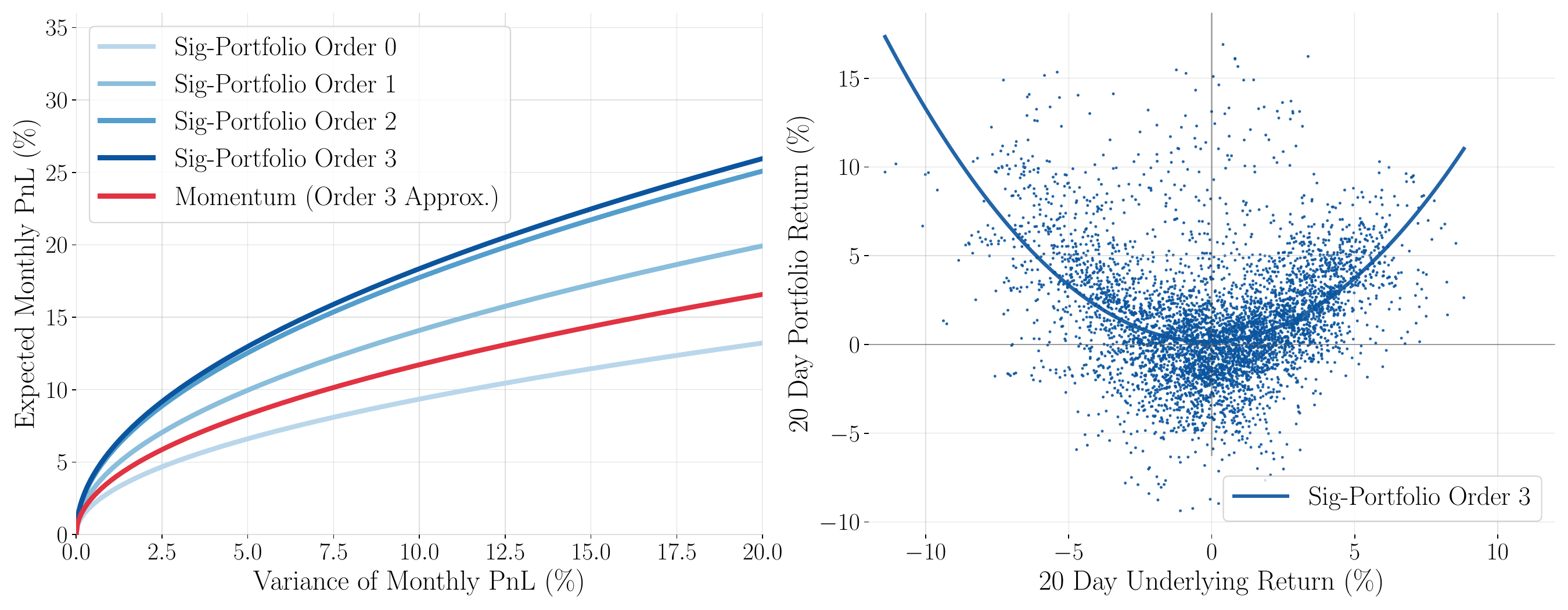}
  \caption{\centering Signature efficient frontier comparison for optimal Sig-Traders and the corresponding learnt 3rd order approximation of the MACD momentum strategy (LHS). The order 3 Sig-Trading strategy return profile compared to the underlying returns (RHS).}
  \label{fig:MACD_EF_convex}
\end{figure}
\noindent characteristics about the underlying dynamics and factor in the pathwise optimisation to achieve a lower monthly variance. The convex return profile of the order 3 Sig-Trader (RHS of Figure~\ref{fig:MACD_EF_convex}) is reminiscent of a trend following strategy that earns positive returns with high conviction in trending (upwards or downwards) markets.
This example shows us that in fact, besides from focusing on an optimal Sig-Trading solution, by re-casting other strategies in the same format, it can inform us better on how our strategies work and the various exposures that may exist. It might be such that a given strategy is largely exposed to specific terms in the signature, which can explain more about its characteristics and associated risks. We can also analyse how ‘far away’ our strategies are to an optimal Sig-Trading solution and hence better optimise for dynamic risks.

\section{Conclusion}

 Path dependencies such as non-Markovian data structures, or time series exhibiting temporal correlation are a frequent phenomenon in financial data. 
 Momentum and mean-reversion (of assets or signals) are two of the purest features of time series’ data that a trader can exploit, and these features are inherently reliant on the whole path. However,  many traditional techniques used for portfolio oprimisation 
 are too inflexible to handle such structural complexities of the data and signals. 
 Signature trading strategies, first developed and covered in detail in the thesis of Perez (\cite{PerezArribas2020SignaturesFinance}), are a versatile representation of any investment strategy, which have been demonstrated to be a powerful tool in the context of pricing, hedging, and optimal execution.
 
 We observe that in fact their advantages can be carried over to more general portfolio optimisation problems as they encompass many common trading styles that are present in practice (including the before-mentioned momentum and mean-reversion). More specifically, in this paper we extend classical factor models  into the Sig-Trading framework, obtaining a closed form solution to the optimal mean-variance Sig-Trading strategy and derive a clear intuition for portfolio managers to navigate and use this formula in a path-dependent context. Furthermore, by lifting the mean-variance optimisation into the lead-lag signature space (See Definition \ref{eq:hoff_defn}), we bypass the necessity for any explicit prediction of returns, which is commonly required in traditional settings. This alleviates the accumulation of asymmetric residuals from the prediction phase, which can often be difficult to control. Moreover, the Sig-Trading framework simultaneously captures the joint signal-asset dynamics, whilst performing a dynamic optimisation which naturally incorporates a drawdown control within the objective function. 

In summary, the Sig-Trading framework provides an alternative to machine learning methods, in its ability to handle path-dependent, non-linear dynamics and signals in a portfolio opitmisation context. Unlike machine learning methods, our framework requires no training or gradient descent optimisation and provides a closed-form solution that is lightweight to work with in practice.  Our closed-form solution is tractable, easily implementable and ensures interpretability of the derived optimal strategies. Overall, our results provide more intuition than ML-based methods and establish more discretion in the fitting procedure than classical methods, providing a malleable framework that still solves many challenges faced when working with financial data.

\begin{appendices}

\section{Rough Path \& Tensor Algebra Preliminaries} \label{sec:appx_rough_paths}

We aim to keep this paper self-contained by recalling the concepts and definitions we explicitly use in this paper. In this appendix we recall necessary fundamental building blocks used in our derivations. 

\subsection{The Tensor Algebra}
In this section, we define the space on which the signature is defined and introduce notations that are used throughout the paper.
\begin{definition}{(Tensor Algebra).}
	Let $d \geq 1$. We define the extended tensor algebra over $\R^d$ by
\begin{align*}
T((\R^d)) := \left\{ \text{  } \textbf{a} = (a_0, a_1, \dots, a_n, \dots) \text{  } \left| \text{  } a_n \in (\R^d)^{\otimes n} \text{  } \right\} \right. 
\end{align*}
Similarly, we define the truncated tensor algebra of order $N \in \N$ and the tensor algebra by $T^{(N)}(\R^d)$ and $T(\R^d)$ respectively, by
\begin{align*}
T^{(N)}(\R^d) := & \left\{ \text{  } \textbf{a} = (a_n)_{n=0}^\infty \text{  } \left|  \text{  } a_n \in (\R^d)^{\otimes n} \text{ and } a_n = 0 \text{  } \forall n \geq N  \text{  } \right\} \right. \subset T((\R^d)) \\
= & \bigoplus_{n = 0}^N (\R^d)^{\otimes n}, \\
T(\R^d) := & \bigoplus_{n = 0}^\infty (\R^d)^{\otimes n} \subset T((\R^d)).
\end{align*}
Note that the truncated tensor algebra of order $N$ has dimension $\sum_{k=0}^N d^k = \frac{d^{N+1}-1}{d-1}$. \par
\end{definition}
\begin{remark}
    Intuitively, we have that the zero-th level of the tensor algebra, $(\R^d)^{\otimes 0}$, is simply the set of all scalars $a_0 \in \R$, with dimension $d^0 = 1$. At the first level, $(\R^d)^{\otimes 1}$ is the set of all $\R$-valued vectors of length $d$, with dimension $d^1$. Likewise, at the second level, $(\R^d)^{\otimes 2}$ is the set of all $\R$-valued $d \times d$ matrices. Then the truncated tensor algebra at order 2, $T^{(2)}(\R^d)$, has dimension $1 + d + d^2$ and contains all $\R^d$-valued tensors of order $0,1,2$.
\end{remark}
\begin{definition}{(Dual Space of the Tensor Algebra).} 
    Let $\{ e_1, \dots , e_d \} \subset \R^d$ be a basis for $\R^d$, then it has a dual basis $\{ e_1^*, \dots , e_d^* \} \subset (\R^d)^*$ for $(\R^d)^*$, the dual space of $\R^d$. Recall that this dual space is the space of all linear functions $\R^d \to \R$. We may similarly define a basis for $T((\R^d))$ and its dual space $T((\R^d)^*)$. \par

    \noindent We identify this dual space of the tensor algebra, $T((\R^d)^*)$, with the space of all words. Consider the following alphabet $A_d:=\{\bluebf{1},\dots,\bluebf{d} \}$, which consists of $d$ letters. In order to ease notations, we make the following identification:
\begin{align}\label{eq:word_functional}
	e_{i_1}^* \otimes \dots \otimes e_{i_n}^* \in T((\R^d)^*) \leftrightarrow \bluebf{i_1} \dots \bluebf{i_n} \in \W(A_d),
\end{align}
where $\W(A_d)$ is the real vector space of all words with the alphabet $A_d$. The empty word will be denoted by \bluebf{$\emptyset$}. We then have the identification $T((\R^d)^*) = \W(A_d).$ That is, that any linear functional  $\ell:T((\R^d)) \to \R$ can be identified via elements in $\W(A_d)$. Hence we can think of $\textit{words}$ as linear functions on the tensor algebra.
\end{definition}
\begin{example}
    Let $\bX \in T((\R^d))$ be an element of the tensor algebra. We can view $\bX$ in terms of its elements within the tensor algebra and each multi-index corresponding to a word, e.g.
    \begin{align*}
        \bX = \left( \underbrace{
        \begin{matrix}
            \text{ } \vspace{0.2cm}  \\
            \bX^{\bluebf{\emptyset}} \\
            \text{ } \vspace{0.2cm}
        \end{matrix}}_{\textstyle \in (\R^d)^{\otimes 0}} , \text{ }
        \underbrace{\begin{pmatrix}
            \bX^{\bluebf{1}}   \\
            \vdots \\
            \bX^{\bluebf{d}}
        \end{pmatrix}}_{\textstyle  \in (\R^d)^{\otimes 1}}
        \text{ } , \text{ }
        \underbrace{\begin{pmatrix}
            \bX^{\bluebf{11}} & \dots &  \bX^{\bluebf{1d}}   \\
            \vdots & \ddots & \vdots \\
            \bX^{\bluebf{d1}} & \dots & \bX^{\bluebf{dd}}
        \end{pmatrix}}_{\textstyle \in (\R^d)^{\otimes 2}} \text{ } , \text{ } \dots
        \right).
    \end{align*}
   The space of all words is defined as 
    \begin{align*}
        \W(A_d) = \{ \bluebf{\emptyset}, \bluebf{1}, \dots, \bluebf{d}, \bluebf{11}, \dots, \bluebf{dd}, \dots \}.
    \end{align*}
\end{example}
\noindent Two algebraic operations on $\W(A_d)$ are the sum and concatenation. The sum of two words $\bluebf{w}$ and $\bluebf{v}$ is just the formal sum $\bluebf{w} + \bluebf{v} \in \W(A_d)$. The concatenation of $\bluebf{w} = \bluebf{i_1} \dots \bluebf{i_n}$, $\bluebf{v} = \bluebf{j_1} \dots \bluebf{j_m} \in \W(A_d)$ is defined by 
\begin{align*}
    \bluebf{wv} :=  \bluebf{i_1} \dots \bluebf{i_n}\bluebf{j_1} \dots \bluebf{j_m} \in \W(A_d).
\end{align*}
With some abuse of notation, we will then use the concatenation on $\W(A_d)$ and $T((\R^d)^*)$ interchangeably, in the sense that we will sometimes write  $\ell \bluebf{w}$ for  $\ell \in T((\R^d)^*), \bluebf{w} \in \W(A_d)$. 

\begin{definition}{(Shuffle Product).} \label{defn:shuffle_product}
    The shuffle product $\shuffle : \W(A_d) \times \W(A_d) \to \W(A_d)$ is defined inductively by
    \begin{center}
        $\bluebf{ua} \shuffle \bluebf{vb} = (\bluebf{u} \shuffle \bluebf{vb})\bluebf{a} + (\bluebf{ua} \shuffle \bluebf{v}) \bluebf{b}$ \par
        $\bluebf{w} \shuffle \bluebf{\emptyset} = \bluebf{\emptyset} \shuffle \bluebf{w} = \bluebf{w} $
    \end{center}
    for all words $\bluebf{u}, \bluebf{v}$ and letters $\bluebf{a}, \bluebf{b} \in \W(A_d)$.
\end{definition}
\begin{example} \label{ex:shuffle_ex}
        Let $\bluebf{w} = \bluebf{12}, \bluebf{v} = \bluebf{34}$, then $\bluebf{w} \shuffle \bluebf{v}$ is given by
        $$
        \bluebf{12} \shuffle \bluebf{34} =  \bluebf{1234} + \bluebf{1324} + \bluebf{1342} + \bluebf{3124} + \bluebf{3142} + \bluebf{3412}
        $$
\end{example}

\begin{remark}
We make extensive use of the fact that polynomials of linear functionals can be expressed as shuffle products of the linear functionals themselves.
\end{remark}

\subsection{Rough Paths \& The Signature }

\begin{definition}{($p$-variation).}
    Let $p \geq 1$ and $X:[t',T] \to \R^d$ be a $d$-dimensional continuous path. We say the $p$-variation of $X$ is denoted as the seminorm
    $$
    \norm{X}_p = \left( \sup_{\mathcal{P}} \sum_{[s,t] \in \mathcal{P}} \norm{X_t - X_s}^p \right)^{\frac{1}{p}}
    $$
    where $\norm{\cdot}$ is any norm on $\R^d$ and the supremum is taken over all partitions $\mathcal{P}$ of the interval $[t',T]$.
\end{definition}

\begin{definition}{(Space of $p$-variation paths).}
    We denote the space of all $\R$-valued $d$-dimensional paths of finite $p$-variation, to be $C^{p-var}([0,T];\R^d)$.
\end{definition}
\begin{remark}
    We refer to the paths of \textit{bounded variation} as elements of the space $C^{1-var}([0,T];\R^d)$. Note that all continuous piecewise smooth paths $X_{0,T} \in C^{1-var}([0,T];\R^d)$.
\end{remark}

\begin{definition}{(Geometric $p$-rough paths).} \label{defn:geom_rough_path}
    We define the space of \textit{geometric $p$-rough paths}, $G^{\lfloor p \rfloor}(\R^d)$, as the closure of the space of signatures of smooth paths, at order $\lfloor p \rfloor$, namely
    $$
    G^{\lfloor p \rfloor}(\R^d) \coloneqq \overline{ \left\{ \text{ } \bX^{\leq N} : \Delta_T \to T^{(N)}(\mathbb{R}^d) \quad \big\vert \quad N = \lfloor p \rfloor \text{ } \right\} }^{d_{p-\textup{var}}}
    $$
    where the closure is with respect to the $p$-variation metric (defined in \cite{Lyons2007DifferentialPaths}, Definition 1.5).
\end{definition}
\begin{theorem}{(Extension Theorem, \cite{Lyons2007DifferentialPaths}, Theorem 3.7).}
    Let $p>1$ be a real number and $\bX^{\leq N}_{s,t} \in G^{\lfloor p \rfloor}(\R^d)$ be the truncated signature at level $N \in \N$ of the path $X_{s,t}\in \X^d_{s,t}$. Then for every $n \geq \lfloor p \rfloor + 1$, there exists a unique $\bX^{n}_{s,t}$ such that
    $$
    (s,t) \mapsto \bX_{s,t} = (1, \bX^{1}_{s,t}, \dots, \bX^{\lfloor p \rfloor}_{s,t}, \dots, \bX^{n}_{s,t}, \dots) \in T((\R^d))
    $$
    has finite $p$-variation. We denote $\bX_{s,t}$ as the \textit{extension} of $\bX^{\leq N}_{s,t}$.
\end{theorem}

\begin{remark}
    The signature of a (piecewise) smooth path $X_{0,T} \in \mathcal{X}^d_{0,T}$ is itself a geometric rough path, namely $\bX_{0,T}^{<\infty} \in G^{\lfloor 1 \rfloor}(\R^d)$.
\end{remark}

\begin{lemma}{(Shuffle product property).}\label{eq:shuffle_product_property}
    Let $\bX^{< \infty}_{0,T} \in G^{[p]}(\R^d)$ be a geometric rough path and let $\ell_1, \ell_2 \in T((\R^d)^*)$ be elements of the dual space of the tensor algebra, then
    \begin{align*}
        \langle \ell_1, \bX^{< \infty}_{0,T} \rangle \langle \ell_2, \bX^{< \infty}_{0,T} \rangle = \langle \ell_1 \shuffle \ell_2 , \bX^{< \infty}_{0,T} \rangle
    \end{align*}
\end{lemma}
\begin{remark}
    The \textit{shuffle product property} is heavily used in deriving an explicit representation of the variance of PnL in Section \ref{sec:original}.
\end{remark}
\begin{proposition}{(Signature is point-seperating, ( \cite{Hambly2010UniquenessGroup}, \cite{Boedihardjo2016TheUniqueness} )).}
    Let $X: [0,T] \to \R^d$, then its signature $\bX_{0,T}^{<\infty} \in G^{[p]}(\R^d)$ is unique up to tree like equivalence and translation.
\end{proposition}
\begin{corollary} \label{eq:add_time}
    Let $\hat{X} : [0,T] \to \R^{d+1}$ be the associated add-time process of $X$. Then its signature $\hat{\bX}_{0,T}^{< \infty}$ uniquely determines $X$ up to translations.
\end{corollary}
\begin{definition}{(Probability measure on the space of paths).}
        Let $(\Omega, \F, (\F_t)_{t \in [0,T]}, \P)$ be a filtered probability space such that $\F=(\F_t)_{t \in [0,T]}$ is the filtration generated by the non-negative $\R^d$-valued stochastic process $X=(X_t)_{t \in [0,T]}$. We say its path trajectories $X_{0,T}$ are sampled under the probability measure $\P$ and we denote $\mathcal{P}(\X^d_{0,T})$ as the set of all such probability measures.
\end{definition}
\begin{definition}{(Expected Signature).}
    Let $X$ be defined as previous. Then for any such $\R^d$-valued random path $X_{0,T}$, we can see that taking its signature $\bX_{0,T}$ is also a random variable under $\P \in \mathcal{P}(\X^d_{0,T})$ and so hence we can define the notion of an \textit{expected signature} under $\P$ at each order as
    $$
    \Ex^{\P} \left[ \bX_{0,T}^n \right] \coloneqq \Ex^{\P} \left[ \text{   } \idotsint\displaylimits_{s < u_1 < \dots < u_k < t} dX_{u_1} \otimes \dots \otimes dX_{u_k} \right] \in (\R^d)^{\otimes n} 
    $$
    where we have
    $$
    \Ex^{\P} \left[ \bX^{< \infty}_{0,T} \right] \coloneqq (1, \Ex^{\P}[\bX_{0,T}^1] , \dots, \Ex ^{\P}[ \bX_{0,T}^n ], \dots ) \in T((\R^d)).
    $$
    
    \noindent The map $\P \mapsto \Ex^{\P} \left[ \bX_{0,T}^{< \infty} \right] \in T((\R^d))$, which maps the probability measure $\P$ to its expected truncated signature, is injective (\cite{Chevyrev2022SignatureProcesses}).
\end{definition}
\clearpage
\begin{remark}
    The \textit{expected signature} allows us to systematically characterize the empirical probability measure on the streams in a model-free sense. The expected signature of paths $X \sim \P$, i.e $\Ex^\P [ \bX^{< \infty}_{0,T} ]$ can be thought of as the moment generating function of a path-valued random variables $X$. 
    We must also note, however, that the assumption that the expected signature even exists at all is a strong assumption for all levels of the signature $N \in \N$. For example, the authors in \cite{Bayer2021OptimalSignatures} highlight that this rules out many stochastic volatility models, such as the Heston model.
\end{remark}

\begin{lemma}{(Factorial Decay).}\label{eq:factorial_decay} \par
The reason we are able to work with the truncated signature with sufficient confidence is due to the factorial decay of the terms of the signature. Let $ X \in \X^d_{0,T} \subset C^{1-var}([0,T],\mathbb{R}^m)$ and $[s,t] \subset [0,T].$ Then $ \forall n \geq 1$ 
\begin{align*}
    \lVert \bX^n_{s,t} \rVert = \left\lVert \quad \idotsint\displaylimits_{t_0 < u_1 < \dots < u_n < t} dX_{u_1} \otimes \dots \otimes dX_{u_n} \right\rVert \leq \frac{(\lVert X \rVert_{1,[s,t]})^n}{n!}
\end{align*}
\end{lemma}

\begin{theorem}{(Universal Approximation, \cite{Levin2013LearningSystem}, Theorem 3.1).} \label{eq:universal_approx} 
Let $K \subset C^{1-var}([0,T],\mathbb{R}^d)$\footnote{$K$ should be a subset of tree-reduced paths. However, as all of the paths we are concerned with are tree-reduced, this is not an issue.} be a compact subset of paths. For any $\phi \in C(K,\R)$, and for every $\epsilon>0$, there exists a linear functional  $\ell \in T((\R^d))^*$ such that
\begin{align*}
\sup_{X \in K} \lVert \phi(X) - \langle \ell, \bX^{< \infty}_{0,T} \rangle \rVert < \epsilon 
\end{align*}
where the choice of suitable candidate topology is discussed in \cite{Cass2022TopologiesSpace}.
\end{theorem}
\begin{remark}
    If we imagine that the optimal adapted dynamic trading strategy $\xi_t$ is simply a function of the past path, i.e $\xi_t: X_{0,t} \mapsto \phi(X_{0,t})$ for some non-linear $\phi$, then Theorem \ref{eq:universal_approx} allows us to approximate the non-linear $\phi$ by $\phi(X_{0,t}) \approx \langle \ell , \bX^{<\infty}_{0,t} \rangle$ where  $\ell$ is linear. This will be made more precise in Section~\ref{sec:market_model}.
\end{remark}
    \noindent \textbf{Notation:}  Throughout, we denote the time-augmented process by $\hat{X}_t=(t, X_t), t \in [0,T]$, the signature of $\hat{X}_{0,T}$ by $\hat{\bX}^{<\infty}_{0,T}$ and the signature of the lead-lag process $\hat{\bX}^{LL,<\infty}_{0,T}$.

\begin{example}{(Linear functional on the signature).} \label{eq:example_sigs} 
Let $X=(X^{1}, \dots, X^{d}) : [0,T] \to \R^d$ be a $d$-dimensional path. We previously made the identification that any linear functional on the tensor algebra can be identified via elements in the space of all words. Take for example:
\begin{enumerate}
    \item Consider any signature term at the second level of the $i$-th and $j$-th component of $X$. 
    \begin{align*}
        \bX^{\bluebf{ij}}_{0,T} = \iint\displaylimits_{0<u_i < u_j <T}  \circ  dX_{u_i}^i  \circ dX_{u_j}^j = \int^T_0 (X_t^i - X_0^i)  \circ  dX_t^j = \langle  \bluebf{ij}, \hat{\bX}^{<\infty}_{0,T} \rangle
    \end{align*}
    where we can think of $\bluebf{ij}$ as a word in the space $\W(A_d)$, which can be identified as a linear functional on the tensor algebra, as shown in \eqref{eq:word_functional}.
    \item Consider an arbitrary linear functional  $\ell \in T((\R^d)^*)$ on the signature, then its (Stratonovich) integral against the $i$-th component of the $d$-dimensional process $X$ is simply the linear functional concatenated with the letter $\bluebf{i}$, applied to the signature $\bX^{<\infty}_{0,T}$, such as:
    \vspace{-0.1cm}
    \begin{align*}
        \int^T_0 \langle \ell, \bX^{<\infty}_{0,t} \rangle  \circ  dX_t^i = \langle \ell \bluebf{i} , \bX^{<\infty}_{0,T} \rangle
    \end{align*}
    \vspace{0.05cm}
    \item Let
    $$
    \ell = \alpha_0 \bluebf{1} + \alpha_1 \bluebf{2} + \alpha_2 \bluebf{12}.
    $$
    When applied to the signature, we obtain
    $$
    \langle \ell, \bX^{<\infty}_{0,t} \rangle = \alpha_0 \bX^{\bluebf{1}}_{0,t} + \alpha_1 \bX^{\bluebf{2}}_{0,t} + \alpha_2 \bX^{\bluebf{12}}_{0,t}
    $$
    and if we consider the concatenation $\ell \bluebf{34}$, then we have
    $$
    \langle \ell \bluebf{34}, \bX^{<\infty}_{0,t} \rangle = \alpha_0 \bX^{\bluebf{134}}_{0,t} + \alpha_1 \bX^{\bluebf{234}}_{0,t} + \alpha_2 \bX^{\bluebf{1234}}_{0,t}.
    $$
\end{enumerate}
\end{example}

\section{Proofs}

\subsection{Proof of Theorem~\ref{thm:int_sig}} \label{sec:proof_of_pnl_thm}

\pnlthm*

\begin{proof}
    Let $\hat{\Z}^{\leq M}$ be the $M$-th order signature of the (time-augmented) market factor process $\hat{Z}$ and define $\hat{Y}_t = (\hat{Z}_t, \hat{Z}_t)$ and its $M$-th order truncated signature as $\hat{\mathbb{Y}}^{\leq M}_t$. In order to prove this theorem, we first state some important results that we shall use as tools throughout. We denote the quadratic co-variation of 2 processes at time $t$ as $[ \cdot , \cdot ]_t$.
    \begin{enumerate}
        \item $\int_0^T \langle l_m, \hat{\Z}^{\leq M}_{0,t} \rangle \circ dX_t^m = \langle l_m \bluebf{f}(m), \hat{\mathbb{Y}}^{\leq M+1}_{0,T} \rangle $
        \item $\int^T_0 Y dX = \int^T_0 Y \circ dX - \frac{1}{2} \left[ X,Y \right]$ for 2 stochastic processes $X, Y$.
        \item $\Z^{LL, \leq 2}_{0,T} = \hat{\mathbb{Y}}^{\leq 2}_{0,T} + \psi_{0,T}$ where
        $$
        \psi_{0,T} = 
        \begin{pmatrix}
            0 & - \frac{1}{2} [ \hat{Y} ]_{0,T}  \\
            \frac{1}{2} [ \hat{Y} ]_{0,T} & 0 \\
        \end{pmatrix}
        $$
        \item $\Z^{LL, \leq M}_{0,T} = \int^T_0 \hat{\Z}^{\leq M-1}_{0,t} \otimes d \hat{Z}_t$ 
        \item $\langle  \bluebf{m} ,  \hat{\Z}^{\leq M}_{0,t}  \rangle = \langle \bluebf{f}(m),  \hat{\mathbb{Y}}^{\leq M}_{0,t} \rangle$ 
        \item $[ \int^T_0 \xi dX, Y] = \int^T_0 \xi d[X,Y] $
    \end{enumerate}

    First, we observe that the trading strategy PnL can be decomposed asset-wise such that $V_T = \sum_{m=1}^d V_T^m$ where $V_T^m$ is the PnL of the trading strategy of the $m$-th asset. Hence, all that needs to be shown is that for arbitrary asset $m$, 
    \begin{align} \label{eq:pnl_by_asset}
    V_T^m = \int^T_0 \langle l_m, \hat{\Z}_{0,t}^{< \infty} \rangle dX_t^m = \langle l_m \bluebf{f}(m), \hat{\Z}^{LL,<\infty}_{0,T} \rangle
    \end{align}
    where the integral above is in the It\^{o} sense. We also wish to show that this result holds for any truncation level $M \geq 1$ and for any number of market factors, $N$. \par
    \noindent Let us fix truncation level $M \geq 1$. By (2) we can decompose the It\^{o} integral in \eqref{eq:pnl_by_asset} into a Stratonovich integral and a quadratic variation correction term, i.e 
    \begin{align}
    \notag \int^T_0 \langle l_m, \hat{\Z}_{0,t}^{\leq M} \rangle dX_t^m & = \int^T_0  \langle l_m, \hat{\Z}_{0,t}^{\leq M} \rangle \circ dX_t^m - \frac{1}{2} \left[  \langle l_m, \hat{\Z}_{0,\cdot}^{\leq M} \rangle , X^m \right]_T \\
     \notag & = \overbrace{\langle l_m \bluebf{f}(m), \hat{\mathbb{Y}}^{\leq M+1}_{0,T} \rangle }^{\textup{by } (1)} \quad \text{   } - \overbrace{\frac{1}{2} \left[ \left\langle l_m , \int^\cdot_0  \hat{\Z}_{0,t}^{\leq M-1} \otimes d \hat{Z}_t \right\rangle , X^m  \right]_T}^{\textup{by } (4)} \\ 
     \notag  & = \langle l_m \bluebf{f}(m), \hat{\mathbb{Y}}^{\leq M+1}_{0,T} \rangle  \quad \text{   } - \overbrace{ \frac{1}{2} 
     \left[ \left\langle l_m \bluebf{f}(m) , \int^\cdot_0  \hat{\mathbb{Y}}_{0,t}^{\leq M-1} \otimes d \hat{Y}_t \right\rangle , X^m \right]_T}^{\text{by } (5)} \\
     \notag  & = \langle l_m \bluebf{f}(m), \hat{\mathbb{Y}}^{\leq M+1}_{0,T} \rangle  \quad \text{   } - \frac{1}{2} 
     \left\langle l_m \bluebf{f}(m) , \left[  \int^\cdot_0  \hat{\mathbb{Y}}_{0,t}^{\leq M-1} \otimes d \hat{Y}_t , X^m \right]_T \right\rangle  \\
    \notag  & = \langle l_m \bluebf{f}(m), \hat{\mathbb{Y}}^{\leq M+1}_{0,T} \rangle  \quad \text{   } - \frac{1}{2} 
     \left\langle l_m \bluebf{f}(m) , \overbrace{\int^\cdot_0  \hat{\mathbb{Y}}_{0,t}^{\leq M-1}  d \left[ \hat{Y}_t , X^m \right]_T}^{\text{by } (6)} \right\rangle  \\
     \notag  & = \langle l_m \bluebf{f}(m), \hat{\mathbb{Y}}^{\leq M+1}_{0,T} \rangle  \quad \text{   } -  \frac{1}{2} \left\langle l_m \bluebf{f}(m), \int^T_0 \hat{\mathbb{Y}}_{0,T}^{\leq M-1} \otimes d [\hat{Y}]_t \right\rangle  \\
     & \label{eq:to_be_shown} = \left\langle l_m \bluebf{f}(m), \hat{\mathbb{Y}}^{\leq M+1}_{0,T} - \frac{1}{2} \int^T_0 \hat{\mathbb{Y}}_{0,t}^{\leq M-1} \otimes d [\hat{Y}]_t \right\rangle.
    \end{align}
    Hence, what remains to be shown is that the RHS of \eqref{eq:pnl_by_asset} is equal to \eqref{eq:to_be_shown}, i.e that
    \begin{align} \label{eq:to_be_shown2}
    \left\langle l_m \bluebf{f}(m), \hat{\Z}^{LL,\leq M+1}_{0,T} \right\rangle = \left\langle l_m \bluebf{f}(m) , \hat{\mathbb{Y}}^{\leq M+1}_{0,T} - \frac{1}{2} \int^T_0 \hat{\mathbb{Y}}_{0,t}^{\leq M-1} \otimes d [\hat{Y}]_t \right\rangle
    \end{align}
    for any truncation level $M \geq 1$. For the case when $M=1$, we have that
    $$
    \hat{\Z}^{LL,\leq 2}_{0,T} = \hat{\mathbb{Y}}^{\leq 2}_{0,T} + \psi_{0,T}
    $$
    as defined in (3), which is proven in \cite{Flint2016DiscretelyProcess}, Theorem 4.1. Therefore, we can clearly see that
    \begin{align*}
        \left\langle l_m \bluebf{f}(m), \hat{\Z}^{LL,\leq 2}_{0,T} \right\rangle = & \left\langle l_m \bluebf{f}(m),  \hat{\mathbb{Y}}^{\leq 2}_{0,T} + \psi_{0,T} \right\rangle \\
        = & \left\langle l_m \bluebf{f}(m),  \hat{\mathbb{Y}}^{\leq 2}_{0,T} - \frac{1}{2} [\hat{Y}]_T  \right\rangle \\
        = & \left\langle l_m \bluebf{f}(m),  \hat{\mathbb{Y}}^{\leq 2}_{0,T} - \frac{1}{2} \int^T_0 1 d [\hat{Y}]_t \right\rangle \\
        = & \left\langle l_m \bluebf{f}(m),  \hat{\mathbb{Y}}^{\leq 2}_{0,T} - \frac{1}{2} \int^T_0 \hat{\mathbb{Y}}^{\leq 0}_{0,T} \otimes  d [\hat{Y}]_t \right\rangle 
    \end{align*}
    Hence, the statement \eqref{eq:to_be_shown2} holds for $M=1$. By the same proof seen in Lemma 3.2.11 in \cite{PerezArribas2020SignaturesFinance}, the result follows for all $M \geq 1$ via induction, taking $I = \bluebf{i_1 i_2 \dots i_k} \in \{1 , \dots, d \}^k$ for the multi-dimensional result. \par
    
    \noindent We note here that by setting $Z = (t, X, f)$, we lose no strength in this argument, since we are still only integrating against $X^m$, for which the result then follows through $\bluebf{f}(m)$ for $m={1,\dots,d}$ and so any extra remaining exogenous factors embedded in $\Z$ do not change the result. Likewise, since out portfolio PnL is defined asset-wise, the result holds for all assets $m$ and so the overall trading strategy PnL is defined as
    \begin{align*}
     V_T = \sum_{m = 1}^d \int^T_0 \langle l_m, \hat{\Z}_{0,s}^{< \infty} \rangle dX_s^m = \sum_{m = 1}^d \langle l_m \bluebf{f}(m), \hat{\Z}^{LL,<\infty}_{0,T} \rangle
    \end{align*}
 \end{proof}


\end{appendices}
\printbibliography[title={References}]

\end{document}